\documentclass[reqno]{amsproc}
\usepackage{tikz}
\usetikzlibrary{arrows,automata,petri,shapes,positioning}
\usepackage{color}
\usepackage{amsmath}
\usepackage[alphabetic]{amsrefs}
\usepackage{amsfonts}
\usepackage{geometry}
\usepackage{amssymb}
\usepackage{bbm}
\usepackage{mathrsfs}
\usepackage{pbox}
\usepackage[utf8x]{inputenc}
\usepackage[tight]{subfigure}
\usepackage{textcomp}
\usepackage{multirow}
\usepackage{paralist}

\definecolor{myurlcolor}{rgb}{0,0,0.4}
\definecolor{mycitecolor}{rgb}{0,0.5,0}
\definecolor{myrefcolor}{rgb}{0.5,0,0}
\usepackage[pagebackref]{hyperref}
\hypersetup{colorlinks,
linkcolor=myrefcolor,
citecolor=mycitecolor,
urlcolor=myurlcolor}

\newcommand{\N}{\mathbb{N}}

\newcommand{\R}{\mathbb{R}}
\renewcommand{\H}{\mathcal{H}}
\newcommand{\B}{\mathcal{B}}

\newcommand{\eq}[1]{~(\ref{#1})}
\renewcommand{\equiv}{\stackrel{\mathrm{def}}{=}}
\newcommand{\eps}{\varepsilon}

\newcommand{\T}[1]{\texttt{#1}}
\newcommand{\eqdef}{\stackrel{\textrm{def}}{=}}

\theoremstyle{plain}
\newtheorem{thm}{Theorem}[section]
\newtheorem{lem}[thm]{Lemma}
\newtheorem{prop}[thm]{Proposition}
\newtheorem{cor}[thm]{Corollary}

\newtheorem*{ulem}{Lemma}
\newtheorem*{uprop}{Proposition}

\newtheorem{defn}[thm]{Definition}
\newtheorem{prob}[thm]{Problem}
\newtheorem*{udefn}{Definition}
\theoremstyle{definition}
\newtheorem{expl}[thm]{Example}

\theoremstyle{remark}
\newtheorem{rem}[thm]{Remark}

\numberwithin{equation}{section}

\renewcommand{\labelenumi}{(\alph{enumi})}
\renewcommand{\theenumi}{(\alph{enumi})}

\newcommand{\beq}{\begin{equation}}
\newcommand{\eeq}{\end{equation}}

\begin{document}



\title{Beyond Bell's Theorem: Correlation Scenarios}

\author{Tobias Fritz}
\address{ICFO--Institut de Ciencies Fotoniques\\ 
Mediterranean Technology Park\\ 
08860 Castelldefels (Barcelona)\\ 
Spain}
\email{tobias.fritz@icfo.es}


\subjclass[2010]{}

\thanks{\textit{Acknowledgements.} The author would like to thank many independent sources for fruitful discussions, most of which were local and realistic, and the EU STREP QCS for financial support. Furthermore, Cyril Branciard, Llu{\'i}s Masanes, Markus M{\"u}ller, Nicolas Gisin, Nihat Ay and Rodrigo Gallego have provided crucial feedback on a draft version of this paper which greatly helped in improving accuracy and presentation.}

\begin{abstract}
Bell's Theorem witnesses that the predictions of quantum theory cannot be reproduced by theories of local hidden variables in which observers can choose their measurements independently of the source. Working out an idea of Branciard, Rosset, Gisin and Pironio, we consider scenarios which feature several sources, but no choice of measurement for the observers. Every Bell scenario can be mapped into such a \emph{correlation scenario}, and Bell's Theorem then discards those local hidden variable theories in which the sources are independent. However, most correlation scenarios do not arise from Bell scenarios, and we describe examples of (quantum) nonlocality in some of these scenarios, while posing many open problems along the way. Some of our scenarios have been considered before by mathematicians in the context of causal inference.
\end{abstract}

\maketitle

%
%

\section{Introduction}
\label{introduction}

\subsection*{Main ideas}

Bell's Theorem~\cites{Bell,Shimony} shows that quantum phenomena cannot be modelled correctly by a theory satisfying the following natural assumptions:

\renewcommand{\labelenumi}{(\Roman{enumi})}
\renewcommand{\theenumi}{(\Roman{enumi})}
\begin{enumerate}
\item\label{R} Realism: Any physical system can be described in terms of a probabilistic mixture of states (=hidden variable values). Composite systems are described by a joint probability distribution over the state spaces of its component systems.
\item\label{L} Locality: Physical systems have spatial components which can be described independently. They do not interact across spacelike separated events.
\item\label{FW} Free will: The parties in a Bell scenario have genuine randomness available which is independent of their environment. 
This is also known as \emph{$\lambda$-independence}~\cite{BY} and as \emph{measurement independence}~\cite{Hall}.
\end{enumerate}

Standard quantum theory fails~\ref{R} due to the way that joint systems are described. It is irrelevant whether~\ref{FW} holds in quantum theory, since~\ref{FW} is only used in combination with~\ref{R} and~\ref{L} in the derivation of the Bell inequalities, which are found to have quantum violations.

In this paper, we are concerned with assumption~\ref{FW}. More precisely, we are actually \emph{not} concerned with~\ref{FW}, since we aim to replace it with a different property:

\renewcommand{\labelenumi}{(\Roman{enumi}')}
\renewcommand{\theenumi}{(\Roman{enumi}')}
\begin{enumerate}
\setcounter{enumi}{2}
\item\label{IS} Independence of sources~\cite{BGP}: if an experiment contains several sources\footnote{It is not perfectly clear to us what ``source'' actually means. One possible definition of source might be that it is a physical system which is, in the \emph{quantum-theoretical} description, independent of its environment: the total initial state should be the tensor product of the system state and an environment state.}, then the theory describes these sources as independent. This means that the joint distribution of hidden variables is a product distribution.
\end{enumerate}

Our observation is that~\ref{FW} becomes obselete when assuming~\ref{IS}, so that one obtains:\\

\noindent\textbf{Bell's Theorem, new version.} \textit{Quantum phenomena cannot be modelled correctly by a theory satisfying~\ref{R},~\ref{L},~\ref{IS}.}\\

Branciard, Rosset, Gisin and Pironio already briefly considered scenarios in which each party has only one measurement setting~\cite{BRGP}*{Sec.~V/VI}. These are a natural continuation of their earlier work~\cite{BGP} which combined~\ref{IS} with~\ref{FW}. Here, we build on their idea and and set up a formal framework for multi-source ``correlation scenarios'' in which each party has only one measurement setting available and derive more results within that framework. There are several advantages to this over the standard approach based on~\ref{FW}:

\begin{itemize}
\item One of the main goals of the hidden variable program was to resurrect a deterministic worldview~\cite{EPR}. However, as has also been observed by 't Hooft~\cite{tHooft} and probably others, determinism is at variance with~\ref{FW} \emph{even without Bell's Theorem} since genuine randomness cannot be created in a deterministic world. This tension between determinism and free will has been known to philosophers long before and led them to seek definitions of human free will compatible with determinism~\cite{McKenna}.
\item Free will is an observer-centric notion which, depending on the theory, may require the observer to live outside that part of the universe described by the theory. In contrast, the property~\ref{IS} concerns only observer-independent physical systems and has clear physical meaning. our formalism is best viewed as devoid of any concious agents.
\item Bell's Theorem is often presented as a statement about theories satisfying realism~\ref{R} and locality~\ref{L} only. \ref{FW} is then tacitly assumed without explicit mention, either because one has failed to notice it as an additional and crucial assumption, or because it may be incorrectly regarded as self-evident. In contrast,~\ref{IS} is more easily understood to be a non-trivial assumption.
\item There has been speculation on the relation between quantum mechanics and free will. Our approach elucidates that this discussion is irrelevant to Bell's Theorem (as is well-known to experts, but possibly not to those just learning about Bell's Theorem and assumption~\ref{FW}).
\end{itemize}

Moreover, our formalism allows the consideration of (quantum) correlations which have no analog in standard Bell scenarios and are genuinely new; see Theorems~\ref{C3thm} and~\ref{C4thm}. Our current results are not sufficient to tell what the meaning or relevance of such new kinds of correlations might be; ultimately, we hope for the development of quantum information protocols utilizing them in ways similar to those taking advantage of quantum correlations in standard Bell scenarios, e.g.~quantum key distribution~\cite{Ekert} or certified randomness generation~\cite{PAal}. Another interesting direction might be to consider analogs of the amplification of free will~\cite{CR} for the amplification of independence of sources.

\subsection*{Inference of common ancestors} Some of the mathematical problems we are going to discuss in this paper have been considered before in a totally different context. There is work by Steudel and Ay~\cite{SA} on the \emph{inference of common ancestors}, which concerns question such as this: given three different languages, under which conditions can one derive the existence of a common antecedent language which influenced all three? Or, given the joint distribution of the prevalence of some diseases in a population, under which conditions can one conclude the existence of a certain preexisting quantity or property (like a genetic defect or a specific diet) having some influence on the occurence of all the diseases considered? This is the question of existence of a \emph{common ancestor} in a Bayesian network model~\cite{Pearl}. A variable in a Bayesian network typically has many ancestors, including itself. One then considers models of the given joint distribution of the observed variables in terms of Bayesian networks, in which each observed variable corresponds to a node, the other nodes represent unobserved variables, and each edge represents a causal link. Then the question is whether one can find such a model without a node which is an ancestor of all the observed variables, or whether such a Bayesian network model necessarily requires such a common ancestor.

For the special case of three observed variables $a$, $b$, $c$, the very general results of~\cite{SA} show that when the single-variable Shannon entropies $H(a)$, $H(b)$, $H(c)$ and the joint entropy $H(abc)$ satisfy the inequality
\beq
\label{SAeiI}
H(a) + H(b) + H(c) > 2 H(abc),
\eeq
then the existence of a common ancestor is necessary. In our example: if the vocabulary of three languages is correlated in such a way that the entropy of the joint distribution is so low that the inequality holds, then there needs to be a common precursor having influenced all three.

We will see that the inference of common ancestors is a special case of our formalism. A byproduct of our results will be an inequality similar to but strictly better than\eq{SAeiI}, for the very particular case of three variables; see\eq{newei}.

\subsection*{Directions of future research} We hope that our ideas will spur new developments in several directions:
\begin{itemize}
\item Further study of classical, quantum and generalized correlations in correlation scenarios. The wealth of open problems we present shows that our results are nothing but a first step towards an understanding of correlation scenarios.
\item What are the philosophical implications of our results? How do~\ref{FW} and~\ref{IS} compare from a philosophy of science perspective?
\item Could our correlation scenarios have any relevance for applications like quantum key distribution?
\end{itemize}
A further generalization of correlation scenarios to scenarios with arbitrary causal structure will be considered in~\cite{FS}. Correlation scenarios are a natural intermediate step between Bell scenearios and the arbitrary causal structure of~\cite{FS}.

\subsection*{Organization of this paper} The interested reader should start with the next subsection on terminology and notation, for otherwise the main text will not be comprehensible. The subsequent main part of the paper in Sections~\ref{tlcs} and~\ref{gentlcs} can be read in a linear way. Section~\ref{tlcs} contains the most important material, namely the conceptual discussion and the examples we have considered so far. Those who do not care too much about abstract generalities may stop reading at any point at which they start losing interest. In particular, reading Section~\ref{gentlcs}, which contains an initial sketch of how an abstract approach to our formalism could look like, is not required for understanding the main ideas. It is supposed to be an attempt at laying the formal basis for future work on the subject.

Due to the high amount of technical detail required for completely rigorous proofs, we restrict ourselves in several cases to the presentation of proof sketches. We hope that these make it clear how completely rigorous proofs can be constructed. In cases where a general rigorous proof or definition involves measure theory, the main text provides the proof or definition for the case of discrete hidden variables; Appendix~\ref{app} then treats the general case of hidden variables defined on arbitrary probability spaces.

Since the subject of this paper is relatively new, many questions remain open. In the main text, we mention a wealth of open problems of various difficulties. We warn the reader that trying to solve them can be quite frustrating; our own experience has been that the intuition we have developed for standard Bell scenarios is sometimes more of a hindrance than an asset. Many of our initially promising ideas have turned out to be misconceived. Those that have eventually worked are based on very different concepts ranging from entropic inequalities (Lemma~\ref{C3ei}) via Hardy-type paradoxes (Theorem~\ref{C4thm}) to Choquet's Theorem (see~\ref{gensepproof}). Nevertheless, we hope that our formalism will develop into an alternative approach to the study of nonlocality and will continue to be studied not only from our mathematical point of view, but also from both the information processing and the philosophical perspective. For example, the recent ``PBR Theorem''~\cite{PBR} also considers hidden variable theories satisfying~\ref{IS} and a comparison to our approach may be interesting.

Finally, Appendix~\ref{app} contains measure-theoretical details concerning the consideration of non-discrete hidden variables. In the main text, all our definitions and proofs are rigorous only for the case of discrete hidden variables; without exception, the same ideas work in the general case, but the technicalities required are so much more laborious and obscure that we relegate them to the appendix.

A follow-up paper~\cite{FS} will present an even more general formalism for device-independent physics in terms of hidden Bayesian networks. It will comprise not only standard Bell scenarios and the formalism we introduce here, but also other scenarios like Popescu's ``hidden'' nonlocality~\cite{Pop}. It will be conceptually similar to hidden Markov models~\cite{LA}.

\subsection*{Terminology and notation}

From now on, we will restrain from using the misleading term \emph{nonlocality} and related terms like \emph{local correlations}. It is misleading terminology insofar as it suggests that nonlocal interactions would be the only way to escape the conclusion of Bell's Theorem; however this is far from correct, since locality is only one of the assumptions~\ref{R},~\ref{L},~\ref{FW}. Moreover, despite the experimental verification of the existence of quantum ``nonlocality''~\cite{AGR}, all known fundamental interactions in physics are of a local nature~\cites{CG,Haag,Jackiw}; see also~\cite{Zeh}. Consequently, we will rather speak of \emph{classical correlations} in analogy with the commonly used term \emph{quantum correlations}. We will use these notions both in the context of standard Bell scenarios as well as in our new \emph{correlation scenarios}.

In the context of our correlation scenarios, we use typewriter-font uppercase letters \texttt{A}, \texttt{B}, \texttt{C}, \ldots{} to enumerate the measurements. Equivalently, one may think of these as observers or parties: since each observer or party gets assigned a fixed measurements which they conduct in each run of the experiment, this is the same. The corresponding measurement outcomes are denoted by lowercase letters $a$, $b$, $c$, \ldots. We denote the joint probability distribution of outcomes of, for example, the joint measurement $(\T{A},\T{Y})$ by $p(a,y)$. This constitutes extensive abuse of notation as it makes expressions like $p(97,-2)$ ambiguous: does this refer to the distribution $p(a,y)$ or to another one like $p(w,z)$? Notwithstanding, we use this notation here in order to keep clutter to a minimum, while making sure that it does not lead to ambiguous expressions. We also keep the order of the variables arbitrary: for example, $p(x,a,b,y)$ stands for the same distribution as $p(a,b,x,y)$, and the one we use depends on which one is more natural in that particular context. Moreover, notation like $p(a,b,x,y)$ makes sense, strictly speaking, only when all variables are discrete; while we do assume that all measurements have only a finite number of possible outcomes, we do not make any discreteness assumption on the hidden variables; see Appendix~\ref{app}.

\subsection*{Necessary background} Any reader looking at this paper will probably already have the necessary understanding of Bell's Theorem~\cites{Bell,Shimony}. Moreover, we also need to assume good familiarity with the notions of (conditional) independence of random variables and conditioning of probabilities. A basic knowledge of the terminology of graphs and hypergraphs is required for Section~\ref{gentlcs}. Some background in Bayesian networks~\cites{KF,Pearl} will be of advantage in order to understand the connection to~\cite{SA}. Reading Appendix~\ref{app} is not possible without some grasp of measure-theoretical probability theory and related subjects.


\section{Examples of correlation scenarios}
\label{tlcs}

In this section, we introduce correlation scenarios by way of example. Using the appropriate dictionary from the standard framework into our formalism, we show how to translate any ordinary Bell scenario as well as the ``bilocality'' scenarios introduced in~\cite{BGP} into a scenario without free will.

We also present the first examples of correlation scenarios, some of which have been considered in~\cite{BRGP} and some of which are new. Obtaining concrete results about these new kinds of correlations has turned out to be difficult; until now, we have been able to do so only by relating to things we were already familiar with (standard Bell scenarios). We hope that future work will show the class of correlation scenarios, as we are going to formally define it in Section~\ref{gentlcs}, to be much richer than what we begin to explore in this paper.

\subsection*{A first example.} 

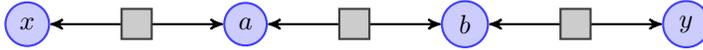
\begin{figure}
\centering
\begin{tikzpicture}[node distance=2.3cm,>=stealth',thick]
\tikzstyle{place}=[circle,thick,draw=blue!75,fill=blue!20,minimum size=4mm]
\tikzstyle{transition}=[rectangle,thick,draw=black!75,fill=black!20,minimum size=4.0mm]
\node[place] (a) at (0,0) {$a$} ;
\node[place] (x) [left=of a] {$x$} ;
\node[place] (b) [right=of a] {$b$} ;
\node[place] (y) [right=of b] {$y$} ;
\draw[<->] (x) -- (a) node [midway,above=-6pt,transition] {} ;
\draw[<->] (a) -- (b) node [midway,above=-6pt,transition] {} ;
\draw[<->] (y) -- (b) node [midway,above=-6pt,transition] {} ;
\end{tikzpicture}
\caption{The correlation scenario $P_4$.} 
\label{P4fig}
\end{figure}

Let us consider an experimental setup as depicted abstractly in Figure~\ref{P4fig}. There are $4$ parties \T{X}, \T{A}, \T{B}, \T{Y} (circles) arranged in a linear way such that any pair of neighboring parties shares a source (square). Each of these three sources sends out, at time $t_{\textrm{emit}}$, one physical system to each adjacent party. As in the case of ordinary Bell scenarios, these two systems are typically correlated; in the classical case, this is shared randomness, while in the quantum case, such a correlation can also be entanglement. The parties receive these systems and each party conducts, at time $t_{\textrm{meas}}>t_{\textrm{emit}}$, a fixed measurement on the system(s) they have received; in the case of \T{A} and \T{B}, who receive two systems each, this will typically be a joint measurement operating on both systems simultaneously. In each run of the experiment, the parties obtain and register outcomes $x$, $a$, $b$, $y$. If the experiment is repeated many times, the parties will notice correlations between these outcomes and determine a joint probability distribution $p(x,a,b,y)$. With the parties as vertices and the sources as edges, Figure~\ref{P4fig} has the structure of the path graph $P_4$, and therefore we will speak of the \emph{$P_4$ scenario}. It has first been studied in~\cite{BRGP}*{Sec.~5}.

Ideally, the timing and the geometry of the experiment should guarantee that the leftmost source cannot causally influence $b$ or $y$ in the time between $t_{\textrm{emit}}$ and $t_{\textrm{meas}}$. Similar causal separation should hold between any other pair of source and measurement which do not share an arrow in Figure~\ref{P4fig}. This ensures the validity of assumption~\ref{L}.

Also, the sources should have been prepared in such a way that the correct quantum-mechanical description of the system will take the joint state of the sources to be a product state, and furthermore such that any correlation between them in a potential hidden variable description should be rendered very implausible. In other words, the experiment should try to guarantee that any hidden variable theory not satisfying~\ref{IS} should be very unreasonable and contrived. This may be achieved, for example, by placing the sources at large spatial separation betwen each other and by using sources which employ different physical mechanisms. But of course, since the past light cones of the sources will always intersect, the requirement~\ref{IS} can never be enforced. It will always be possible to explain all observations by, for example, a superdeterministic theory in which everything is predetermined since the beginning of the universe; compare~\cite{SUS}*{Ch.~12}.

As has already been noticed in~\cite{BGP}, this discussion is completely analogous to the discussion of the validity of property~\ref{FW}: there exist hidden variable theories, like superdeterminism, which do not allow free will and therefore evade the conclusion of Bell's Theorem. However, these are generally so contrived that one cannot regard them as scientific theories of physics. Exactly the same applies to our assumption~\ref{IS} in a suitably conducted experiment.

Now we imagine that many runs of such an experiment have been conducted and we are given the joint outcome statistics $p(x,a,b,y)$. In the following, we work with the ideal case of infinite statistics, so that the outcome probabilities $p(x,a,b,y)$ are known with perfect precision.

Then, due to the causal structure of the experiment, one should find that the outcome $x$ is independent of $y$, since \T{X} and \T{Y} do not connect to a common source. Similarly, $x$ should be independent of $b$; in fact, $x$ should be independent of the pair $(b,y)$. Similarly, $y$ should be indepdendent of the pair $(x,a)$. Checking whether this is indeed the case amounts to a consistency check for the experiment.

More formally, these requirements mean that $p(x,a,b,y)$ should be a \emph{correlation}:

\begin{defn}
\label{P4corrdef}
A \emph{correlation} $p$ in the $P_4$ scenario is a distribution $p(x,a,b,y)$ whose marginals factorize as
\beq
\label{P4corrdefeq}
p(x,a,y) = p(x,a)p(y),\qquad p(x,b,y) = p(x) p(b,y).
\eeq
\end{defn}

Any of these two equations implies $p(x,y) = p(x) p(y)$. Upon using this, one finds that\eq{P4corrdef} is equivalent to $p(a|x,y)=p(a|x)$ and $p(b|x,y)=p(b|y)$ for all those values of $x$ and $y$ for which $p(x)>0$ and $p(y)>0$. Upon reinterpreting $x$ and $y$ as settings in a bipartite Bell scenario having outcomes $a$ and $b$, these are the no-signaling equations. However, conceptually,\eq{P4corrdef} has nothing to do with the impossibility of communication between the parties: these cannot do anything else than apply their fixed measurement in each run of the experiment, which renders the very notion of communication meaningless.

We now ask under which conditions a given correlation $p(x,a,b,y)$ is \emph{classical}, i.e.~consistent with the assumptions~\ref{R},~\ref{L},~\ref{IS}. What would it mean to have such a model? Due to~\ref{R}, the state of the systems sent out by each of the three sources can be described in terms of a classical random variable; we will denote these ``hidden'' variables by $\lambda_{\T{X}\T{A}}$, $\lambda_{\T{A}\T{B}}$, $\lambda_{\T{B}\T{Y}}$, respectively, where the index specifies the source which the hidden variable models. For the precise definition of hidden variable, see~\ref{wihv}.

Assumption~\ref{IS} now means that the joint distribution of these hidden variables is a product distribution:
$$
p(\lambda_{\T{X}\T{A}},\lambda_{\T{A}\T{B}},\lambda_{\T{B}\T{Y}}) = p(\lambda_{\T{X}\T{A}}) p(\lambda_{\T{A}\T{B}}) p(\lambda_{\T{B}\T{Y}}) .
$$
A sensible hidden variable model should also satisfy \emph{locality}~\ref{L}: each outcome should be a (deterministic or probabilistic) function of the hidden variables associated to the sources it interacts with \emph{and no others}.

If such a hidden variable model exists for the correlation $p$, then we call $p$ \emph{classical}. A more precise statement is this:

\begin{defn}[\cite{BRGP}]
\label{P4classdef}
A correlation $p(x,a,b,y)$ is classical in the $P_4$ scenario if and only if it can be written in the form
\beq
\label{P4lhv}
p(x,a,b,y) = \int_{\lambda_{\T{X}\T{A}},\lambda_{\T{A}\T{B}},\lambda_{\T{B}\T{Y}}} p(x|\lambda_{\T{X}\T{A}}) p(\lambda_{\T{X}\T{A}}) p(a|\lambda_{\T{X}\T{A}},\lambda_{\T{A}\T{B}}) p(\lambda_{\T{A}\T{B}}) p(b|\lambda_{\T{A}\T{B}},\lambda_{\T{B}\T{Y}}) p(\lambda_{\T{B}\T{Y}}) p(y|\lambda_{\T{B}\T{Y}}) 
\eeq
for some collection of (conditional) distributions
$$
p(x|\lambda_{\T{X}\T{A}}),\quad p(\lambda_{\T{X}\T{A}}),\quad p(a|\lambda_{\T{X}\T{A}},\lambda_{\T{A}\T{B}}),\quad p(\lambda_{\T{A}\T{B}}),\quad p(b|\lambda_{\T{A}\T{B}},\lambda_{\T{B}\T{Y}}),\quad p(\lambda_{\T{B}\T{Y}}),\quad p(y|\lambda_{\T{B}\T{Y}}).
$$
\end{defn}

See~\ref{appcond} for an explanation of what these conditional distributions mean in case that the hidden variables are not all discrete. 

We take this to be a \emph{definition} instead of a proposition or theorem because it is the first time that we have formalized the notion of classical model in a mathematically rigorous way. The representation\eq{P4lhv} can be informally derived from hypotheses~\ref{R},~\ref{L},~\ref{IS} as follows. Applying~\ref{R} and the definition of conditional probability gives
$$
p(x,a,b,y) = \int_{\lambda_{\T{X}\T{A}},\lambda_{\T{A}\T{B}},\lambda_{\T{B}\T{Y}}} p(x,a,b,y|\lambda_{\T{X}\T{A}},\lambda_{\T{A}\T{B}},\lambda_{\T{B}\T{Y}}) p(\lambda_{\T{X}\T{A}},\lambda_{\T{A}\T{B}},\lambda_{\T{B}\T{Y}}).
$$
By locality~\ref{L}, the first factor in the integrand can be replaced by
$$
p(a,b,x,y|\lambda_{\T{X}\T{A}},\lambda_{\T{A}\T{B}},\lambda_{\T{B}\T{Y}}) = p(x|\lambda_{\T{X}\T{A}}) p(a|\lambda_{\T{X}\T{A}},\lambda_{\T{A}\T{B}}) p(b|\lambda_{\T{A}\T{B}},\lambda_{\T{B}\T{Y}}) p(y|\lambda_{\T{B}\T{Y}})
$$
while independence of sources~\ref{IS} guarantees that the second factor is equal to
$$
p(\lambda_{\T{X}\T{A}},\lambda_{\T{A}\T{B}},\lambda_{\T{B}\T{Y}}) = p(\lambda_{\T{X}\T{A}}) p(\lambda_{\T{A}\T{B}}) p(\lambda_{\T{B}\T{Y}}) ,
$$
and then~(\ref{P4lhv}) directly follows.

\begin{rem}
\label{wlogdet}
In the representation\eq{P4lhv}, it can be assumed without loss of generality that the four conditional distributions on the right-hand side are in fact deterministic, i.e.~it can be assumed that the outcomes are functions
$$
x=x(\lambda_{\T{X}\T{A}}),\quad a=a(\lambda_{\T{X}\T{A}},\lambda_{\T{A}\T{B}}),\quad b=b(\lambda_{\T{A}\T{B}},\lambda_{\T{B}\T{Y}}),\quad y=y(\lambda_{\T{B}\T{Y}}).
$$
In the case of discrete hidden variables, this can be seen as follows: if, for example, $a$ is a probabilistic function of $\lambda_{\T{X}\T{A}}$ and $\lambda_{\T{A}\T{B}}$, then the computation of this function can be regarded as the \emph{deterministic} computation taking the values $\lambda_{\T{X}\T{A}},\lambda_{\T{A}\T{B}}$ and an additional random number $r_{\T{A}}\in[0,1]$ as input, calculating $p(a|\lambda_{\T{X}\T{A}},\lambda_{\T{A}\T{B}})$ for each outcome $a$, and then using $r_{\T{A}}$ to determine which one of these finitely many outcomes occurs. But now we can redefine the hidden variable $\lambda_{\T{X}\T{A}}$ to be the pair $\lambda_{\T{X}\T{A}}'=(\lambda_{\T{X}\T{A}},r_{\T{A}})$ which contains the information about the original $\lambda_{\T{X}\T{A}}$ as well as the additional random number $r_{\T{A}}$ required in the computation; the party $\T{X}$ will then also receive this new component of $\lambda_{\T{X}\T{A}}'$, but can just ignore it. In this way, the function $a(\lambda_{\T{X}\T{A}}',\lambda_{\T{A}\T{B}})$ has become deterministic.

Upon applying this kind of hidden variable redefinition for each party, all the outcomes become deterministic functions of the hidden variables.

This reasoning not only applies to $P_4$, but in exactly the same way to any correlation scenario. We will make use of this in the proof of Theorem~\ref{C4thm}. See~\ref{appwlogdet} for a rigorous and general version of this argument.
\end{rem}


It is also not difficult to define what \emph{quantum correlations} are. Informally speaking, a quantum correlation is a correlation $p(x,a,b,y)$ which can be modelled in terms of quantum resources: a bipartite quantum state for each source together with one measurement for each party operating jointly on all the systems received by that party. The Hilbert space dimension of the quantum systems can be arbitrary and will be infinite in general. We take the definition of quantum correlation to be sufficiently obvious that we need to go into detail here; see Definition~\ref{gendefq} for the technicalities.

The following theorem makes the connection to bipartite Bell scenarios. Its first part has also appeared in~\cite{BRGP}.

\renewcommand{\labelenumi}{(\arabic{enumi})}
\renewcommand{\theenumi}{(\arabic{enumi})}
\begin{thm}
\label{P4thm}
\begin{enumerate}
\item\label{P4cl} A correlation $p(a,b,x,y)$ is \emph{classical} in $P_4$ if and only if the associated conditional distribution $p(a,b|x,y)$ is classical in the Bell scenario sense.
\item\label{P4qu} A correlation $p(a,b,x,y)$ is \emph{quantum} in $P_4$ if and only if the associated conditional distribution $p(a,b|x,y)$ is quantum in the Bell scenario sense.
\end{enumerate}
\end{thm}

Note that the use of conditional probabilities here, or in any other context, does not require any particular causal structure among the variables involved.

In forming $p(a,b|x,y)$, it is implicitly assumed that all outcomes for $x$ and $y$ have strictly positive probability; this can always be achieved by redefining the set of outcomes to consist of only those values which occur with positive probability.

Thus, we can roughly summarize our present results as follows: by Definition~\ref{P4classdef}, a correlation $p(a,b,x,y)$ can be interpreted in a conventional bipartite Bell scenario as a no-signaling box together with a specification of input distributions $p(x)$ and $p(y)$; and the correlation is classical (resp.~quantum) if and only if the associated no-signaling box is classical (resp.~quantum).

\begin{proof}[Proof of Theorem~\ref{P4thm}]
\begin{asparaenum}
\item Suppose that $p$ is classical. Then
$$
p(a,b|x,y) = \int_{\lambda_{\T{A}\T{B}}} p(a,b|x,y,\lambda_{\T{A}\T{B}}) p(\lambda_{\T{A}\T{B}}) .
$$
By the assumption\eq{P4lhv}, upon conditioning on $\lambda_{\T{A}\T{B}}$, the variables $(a,x)$ are independent of the variables $(b,y)$; therefore, $p(a,b|x,y,\lambda_{\T{A}\T{B}})=p(a|x,\lambda_{\T{A}\T{B}})p(b|y,\lambda_{\T{A}\T{B}})$, and
\beq
\label{P4hv}
p(a,b|x,y) = \int_{\lambda_{\T{A}\T{B}}} p(a|x,\lambda_{\T{A}\T{B}}) p(b|y,\lambda_{\T{A}\T{B}}) p(\lambda_{\T{A}\T{B}}) .
\eeq
This is the standard representation of the conditional probabilities obtained from local hidden variables in a bipartite Bell scenario. In particular, $p(a,b|x,y)$ will have to satisfy all Bell inequalities.

Conversely, we start from a correlation $p(a,b,x,y)$ for which $p(a,b|x,y)$ satisfies all Bell inequalities. This means in particular that there is a hidden variable $\lambda$ such that
$$
p(a,b|x,y) = \int_{\lambda} p(a|x,\lambda)p(b|y,\lambda)p(\lambda)
$$
Defining $\lambda_{\T{A}\T{X}}=x$, $\lambda_{\T{B}\T{Y}}=y$ and $\lambda_{\T{A}\T{B}}=\lambda$ now yields a hidden variable model in the $P_4$ correlation scenario, i.e.~the right-hand side of~(\ref{P4hv}).
\item Suppose that $p(a,b,x,y)$ is quantum. Then one has one bipartite quantum state at each source and one quantum measurement at each party. We think of the measurement \T{X} as remotely preparing, via steering depending on the outcome $x$, a quantum system for \T{A}. In order to ease notation, we may assume, without loss of generality, the shared state to be pure and \T{X}'s measurement to be projective. Furthermore, we may take \T{X}'s projective measurement to be nondegenerate; going to a degenerate measurement amounts to a coarse-graining of $\T{X}$, which preserves the quantum-mechanical realizability of $p(a,b,x,y)$. By these assumptions, the steered states for \T{A} are a family $\{|\chi_x\rangle\}$ of pure states. Using the same assumptions for \T{Y}, we end up with a family $\{|\mu_y\rangle\}$ of pure steered states for \T{B}.

We now replace the source between \T{X} and \T{A} by a hidden variable defined to be $\lambda_{\T{A}\T{X}}=x$; then the new measurement protocol of \T{X} simply consists in announcing $\lambda_{\T{A}\T{X}}$'s value as his outcome. The new protocol of \T{A} consists in receving $\lambda_{\T{A}\T{X}}$, preparing the quantum state which \T{X} would have steered to given the outcome $\lambda_{\T{A}\T{X}}$, and then proceeding with the measurement specified in the original protocol. This replacement preserves the overall correlation $p(a,b,x,y)$. The same procedure can be applied in order to replace the source between \T{Y} and \T{B} by a hidden variable $\lambda_{\T{B}\T{Y}}$ and the measurement of \T{Y} by the protocol of simply announcing $\lambda_{\T{B}\T{Y}}$'s value as the outcome $y$.

Let $\{A_a\}$ (resp.~$\{B_b\}$) denote the POVM employed by \T{A} (resp.~\T{B}). Then
\beq
\label{P4quantum}
p(a,b|x,y) = \left( \langle\chi_x|\otimes \langle\psi| \otimes \langle\mu_y| \right) \left(A_a\otimes B_b\right)  \left( |\chi_x\rangle\otimes |\psi\rangle \otimes |\mu_y\rangle \right) ,
\eeq
or, in graphical notation~\cite{Coecke},
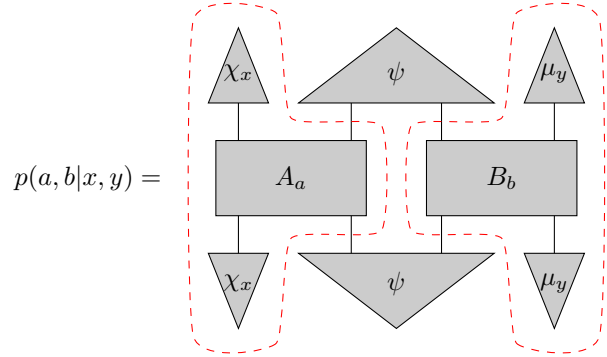
$$
\begin{tikzpicture}
\node at (-4.1,0) {$p(a,b|x,y) = $} ;
\draw[black,fill=black!20] (.4,-.5) rectangle (2.4,.5) ;
\draw[black,fill=black!20] (-.4,-.5) rectangle (-2.4,.5) ;
\node at (1.4,0) {$B_b$} ;
\node at (-1.4,0) {$A_a$} ;
\draw[black,fill=black!20] (-1.3,-1) -- (1.3,-1) -- (0,-2) -- cycle;
\node at (0,-1.4) {$\psi$} ;
\draw[black,fill=black!20] (-1.3,1) -- (1.3,1) -- (0,2) -- cycle;
\node at (0,1.4) {$\psi$} ;
\draw[black,fill=black!20] (-1.7,-1) -- (-2.5,-1) -- (-2.1,-2) -- cycle;
\node at (-2.1,-1.4) {$\chi_x$} ;
\draw[black,fill=black!20] (-1.7,1) -- (-2.5,1) -- (-2.1,2) -- cycle;
\node at (-2.1,1.4) {$\chi_x$} ;
\draw[black,fill=black!20] (1.7,1) -- (2.5,1) -- (2.1,2) -- cycle;
\node at (2.1,1.4) {$\mu_y$} ;
\draw[black,fill=black!20] (1.7,-1) -- (2.5,-1) -- (2.1,-2) -- cycle;
\node at (2.1,-1.4) {$\mu_y$} ;
\draw (2.1,-1) -- (2.1,-.5) ;
\draw (2.1,1) -- (2.1,.5) ;
\draw (-2.1,-1) -- (-2.1,-.5) ;
\draw (-2.1,1) -- (-2.1,.5) ;
\draw (.6,1) -- (.6,.5) ;
\draw (.6,-1) -- (.6,-.5) ;
\draw (-.6,1) -- (-.6,.5) ;
\draw (-.6,-1) -- (-.6,-.5) ;
\draw[red,dashed] plot [smooth cycle,tension=.3] coordinates { (-1.4,-0.8) (-1.6,-2.2) (-2.7,-2.1) (-2.7,2.1) (-1.6,2.2) (-1.4,0.8) (-0.2,0.7) (-0.2,-0.7) } ;
\draw[red,dashed] plot [smooth cycle,tension=.3] coordinates { (1.4,-0.8) (1.6,-2.2) (2.7,-2.1) (2.7,2.1) (1.6,2.2) (1.4,0.8) (0.2,0.7) (0.2,-0.7) } ;
\end{tikzpicture}
$$
Here, the dashed line indicates how to consider $A^x_a\equiv\langle\chi_x|A_a|\chi_x\rangle$ as well as $B^y_b\equiv\langle\mu_y|B_b|\mu_y\rangle$ as operators acting on one part of the bipartite state $|\psi\rangle$. By $\sum_a A_a=\mathbbm{1}$ and normalization of $|\chi_x\rangle$, it follows that $\sum_a A^x_a=\mathbbm{1}$ for all $x$; similarly, $\sum_y B^y_b=\mathbbm{1}$ for all $y$. By definition,~(\ref{P4quantum}) can then be written as
\beq
\label{Bellquantum}
p(a,b|x,y) = \langle\psi|A^x_a\otimes B^y_b|\psi\rangle.
\eeq
This is desired quantum representation of $p$ in a bipartite Bell scenario.

Conversely, we start from a correlation $p(a,b,x,y)$ of the form~(\ref{Bellquantum}). As sources between \T{A} and \T{X} and between \T{B} and \T{Y}, we again take hidden variables defined by $\lambda_{\T{X}\T{A}}=x$ and $\lambda_{\T{B}\T{Y}}=y$; again, the protocol of \T{X} and \T{Y} is simply to announce the values of these variables as their outcome. Only the source between \T{A} and \T{B} is taken to be quantum and produces the bipartite state $|\psi\rangle$ of~(\ref{Bellquantum}). The measurement protocol conducted by \T{A} is similar to above: measure $\lambda_{\T{X}\T{A}}$, use the result as the choice of setting for the subsequent measurement on $|\psi\rangle$, and then announce both outcomes as the total outcome. This protocol can be interpreted as measuring a single POVM given by
$$
|x\rangle\langle x|\otimes A^x_a \mapsto (x,a),
$$
where the left-hand side is a POVM element indexed by $x$ and $a$, and the right-hand side denotes the resulting outcome announced by \T{A}. The analogous POVM is measured by \T{B}. By construction, this reproduces both the desired conditional distribution~(\ref{Bellquantum}) and the marginal distribution $p(x,y)=p(x)p(y)$, and therefore also the whole distribution $p(x,a,b,y)$.
\end{asparaenum}
\end{proof}

\begin{cor}
\begin{enumerate}
\item There exist non-classical quantum correlations in $P_4$.
\item There exist non-quantum correlations in $P_4$.
\end{enumerate}
\end{cor}

\begin{proof}
This follows from the existence of Bell inequality violations and no-signaling violations of Tsirelson's bound~\cite{Tsirelson}, respectively.
\end{proof}

\begin{rem}
\label{notconvex}
Due to Theorem~\ref{P4thm}, we can regard $P_4$ as the analogue of a bipartite Bell scenario within our formalism. Nevertheless, there are several important differences. For one, the correlations live in completely different spaces: in a Bell scenario, one works in the space of conditional distributions $p(a,b|x,y)$, which results in the convexity of the sets of classical and quantum correlations. In contrast, in the case of $P_4$, we work on the level of unconditional distributions $p(x,a,b,y)$, which contain, from the point of view of Bell scenarios, also the information about the distributions of settings $p(x)$ and $p(y)$. The sets of classical and quantum correlations in this formulation are not convex, which can be seen as follows: first, the set of classical correlations contains all the deterministic distributions $p(x,a,b,y)$ in which all measurements always produce the same outcome. Second, \emph{any} probability distribution $p(x,a,b,y)$, and in particular every correlation, is a convex combination of deterministic ones. Third, not every correlation is classical. Thus, not every convex combination of classical correlations is a classical correlation; for that matter, most convex combinations of classical correlations are not even correlations! The same reasoning shows that the set of quantum correlations is not convex. Analogous arguments apply to any other correlation scenario in which non-classical (resp.~non-quantum) correlations exist.
\end{rem}




\subsection*{The scenario $P_5$}

We proceed to the second example of a correlation scenario. It is depicted in Figure~\ref{P5fig}. With parties as vertices and sources as edges, this is the path graph $P_5$, and therefore we will speak of the \emph{$P_5$ scenario}; the conceptual discussion we gave of the $P_4$ scenario applies here and to all following examples just as well. We will see that the $P_5$ scenario relates to the ``bilocality'' scenarios of Branciard, Gisin and Pironio~\cite{BGP} (\emph{BGP scenarios}) just as we have seen the $P_4$ scenario to relate standard bipartite Bell scenarios.

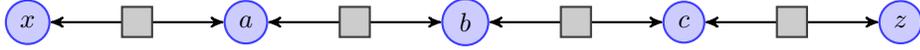
\begin{figure}
\begin{tikzpicture}[node distance=2.3cm,>=stealth',thick]
\tikzstyle{place}=[circle,thick,draw=blue!75,fill=blue!20,minimum size=4mm]
\tikzstyle{transition}=[rectangle,thick,draw=black!75,fill=black!20,minimum size=4.0mm]
\node[place] (a) at (0,0) {$a$} ;
\node[place] (x) [left=of a] {$x$} ;
\node[place] (b) [right=of a] {$b$} ;
\node[place] (c) [right=of b] {$c$} ;
\node[place] (z) [right=of c] {$z$} ;
\draw[<->] (x) -- (a) node [midway,above=-6pt,transition] {} ;
\draw[<->] (a) -- (b) node [midway,above=-6pt,transition] {} ;
\draw[<->] (b) -- (c) node [midway,above=-6pt,transition] {} ;
\draw[<->] (z) -- (c) node [midway,above=-6pt,transition] {} ;
\end{tikzpicture}
\caption{The correlation scenario $P_5$.}
\label{P5fig}
\end{figure}

Given the $5$-variable distribution $p(x,a,b,c,z)$, under which conditions would we expect it to arise from a configuration like Figure~\ref{P5fig}? In other words, what is the analogue of Definition~\ref{P4corrdef}? Following reasoning analogous to the $P_4$ case, the answer is straightforward:

\begin{defn}
\label{P5corrdef}
A \emph{correlation} $p$ in the $P_5$ scenario is a distribution $p(x,a,b,c,z)$ whose marginals factorize as
$$
p(x,a,b,z) = p(x,a,b) p(z),\qquad p(x,a,c,z) = p(x,a) p(c,z),\qquad p(x,b,c,z) = p(x) p(b,c,z).
$$
\end{defn}

Any of these three equations implies $p(x,z)=p(x)p(z)$. Upon using this, the first and third condition can also be written as $\sum_c p(a,b,c|x,z)=\sum_c p(a,b,c|x)$ and $\sum_a p(a,b,c|x,z) = \sum_a p(a,b,c|z)$, respectively, which are formally identical to the no-signaling equations of the BGP scenario. Similarly, the second condition is then equivalent to $p(a,c|x,z)=p(a|x)p(c|z)$, which is also formally identical to a consistency constraint in the BGP scenario~\cite{CF}.

The classicality assumptions~\ref{R},~\ref{L} and~\ref{IS} now yield the following characterization:

\begin{defn}
\label{P5classdef}
A correlation $p(x,a,b,c,z)$ is classical in the $P_5$ scenario if and only if it can be written in the form
\begin{align}
\begin{split}
\label{P5class}
p(&a,b ,c,x,z) \\
& = \int_{\lambda_{\T{X}\T{A}},\lambda_{\T{A}\T{B}},\lambda_{\T{B}\T{C}},\lambda_{\T{C}\T{Z}}} p(x|\lambda_{\T{X}\T{A}}) p(\lambda_{\T{X}\T{A}}) p(a|\lambda_{\T{X}\T{A}},\lambda_{\T{A}\T{B}}) p(\lambda_{\T{A}\T{B}}) p(b|\lambda_{\T{A}\T{B}},\lambda_{\T{B}\T{C}}) p(\lambda_{\T{B}\T{C}}) p(c|\lambda_{\T{B}\T{C}},\lambda_{\T{C}\T{Z}}) p(\lambda_{\T{C}\T{Z}}) p(z|\lambda_{\T{C}\T{Z}})
\end{split}
\end{align}
for some collection of (conditional) distributions 
\begin{align*}
 p(x|\lambda_{\T{X}\T{A}}&),\quad p(a|\lambda_{\T{X}\T{A}},\lambda_{\T{A}\T{B}}),\quad p(b|\lambda_{\T{A}\T{B}},\lambda_{\T{B}\T{C}}),\quad p(c|\lambda_{\T{B}\T{C}},\lambda_{\T{C}\T{Z}}),\quad p(z|\lambda_{\T{C}\T{Z}}).\\[.2cm]
& p(\lambda_{\T{X}\T{A}}),\qquad\quad\:\:  p(\lambda_{\T{A}\T{B}}),\qquad\quad\:\: p(\lambda_{\T{B}\T{C}}),\qquad\quad\:\: p(\lambda_{\T{C}\T{Z}}) .
\end{align*}
\end{defn}

As before, we regard the analogous definition of quantum correlations as straightforward and refer to~\ref{gendefq} for the details.

\begin{thm}
\label{P5thm}
\begin{enumerate}
\item\label{P5cl} A correlation $p(a,b,c,x,z)$ is \emph{classical} in $P_5$ if and only if the associated conditional distribution $p(a,b,c|x,z)$ is classical in the BGP scenario sense.
\item\label{P5qu} A correlation $p(a,b,c,x,z)$ is \emph{quantum} in $P_5$ if and only if the associated conditional distribution $p(a,b,c|x,z)$ is quantum in the BGP scenario sense.
\end{enumerate}
\end{thm}

We abbreviate the proof a bit because it is completely analogous to the proof of Theorem~\ref{P4thm}.

\begin{proof}
\begin{asparaenum}
\item Suppose that $p$ is classical, i.e.~can be written in the form~(\ref{P5class}). Then, 
$$
p(a,b,c|x,z) = \int_{\lambda_{\T{A}\T{B}},\lambda_{\T{B}\T{C}}} p(a,b,c|x,z,\lambda_{\T{A}\T{B}},\lambda_{\T{B}\T{C}}) p(\lambda_{\T{A}\T{B}}) p(\lambda_{\T{B}\T{C}}) .
$$
Upon conditioning on $\lambda_{\T{A}\T{B}}$ and $\lambda_{\T{B}\T{C}}$, we have
$$
p(a,b,c|x,z,\lambda_{\T{A}\T{B}},\lambda_{\T{B}\T{C}}) = p(a|x,\lambda_{\T{A}\T{B}}) p(b|\lambda_{\T{A}\T{B}},\lambda_{\T{B}\T{C}}) p(c|\lambda_{\T{B}\T{C}},z),
$$
and therefore,
$$
p(a,b,c|x,z) = \int_{\lambda_{\T{A}\T{B}},\lambda_{\T{B}\T{C}}} p(a|x,\lambda_{\T{A}\T{B}}) p(b|\lambda_{\T{A}\T{B}},\lambda_{\T{B}\T{C}}) p(c|\lambda_{\T{B}\T{C}},z) p(\lambda_{\T{A}\T{B}}) p(\lambda_{\T{B}\T{C}}),
$$
which is the standard representation of a classical correlation in the BPG scenario~\cite{BGP}. Conversely, upon starting from such a representation, one can again take $\lambda_{\T{X}\T{A}}=x$ and $\lambda_{\T{B}\T{Y}}=y$, and~(\ref{P5class}) also holds.
\item We start with a quantum correlation $p(a,b,c|x,z)$. Upon applying the same steering argument as in the proof of Theorem~\ref{P4thm}, we may assume, in the obvious notation,
\begin{align}
\begin{split}
\label{P5quantum}
p(a,b,& c|x,z) = \\
& \left( \langle\chi_x|\otimes \langle\psi_{\T{AB}}| \otimes \langle\psi_{\T{BC}}| \otimes \langle\mu_y| \right) \left(A_a\otimes B_b\otimes C_c\right)  \left( |\chi_x\rangle\otimes |\psi_{\T{AB}}\rangle \otimes |\psi_{\T{BC}}\rangle \otimes |\mu_y\rangle \right) ,
\end{split}
\end{align}
$$
\begin{tikzpicture}
\node at (-4.1,0) {$p(a,b,c|x,z) = $} ;
\draw[black,fill=black!20] (3.2,-.5) rectangle (5.2,.5) ;
\draw[black,fill=black!20] (.4,-.5) rectangle (2.4,.5) ;
\draw[black,fill=black!20] (-.4,-.5) rectangle (-2.4,.5) ;
\node at (4.2,0) {$C_c$} ;
\node at (1.4,0) {$B_b$} ;
\node at (-1.4,0) {$A_a$} ;
\draw[black,fill=black!20] (-1.3,-1) -- (1.3,-1) -- (0,-2) -- cycle;
\node at (0,-1.4) {$\psi_{\T{AB}}$} ;
\draw[black,fill=black!20] (-1.3,1) -- (1.3,1) -- (0,2) -- cycle;
\node at (0,1.4) {$\psi_{\T{AB}}$} ;
\draw[black,fill=black!20] (1.5,-1) -- (4.1,-1) -- (2.8,-2) -- cycle;
\node at (2.8,-1.4) {$\psi_{\T{BC}}$} ;
\draw[black,fill=black!20] (1.5,1) -- (4.1,1) -- (2.8,2) -- cycle;
\node at (2.8,1.4) {$\psi_{\T{BC}}$} ;
\draw[black,fill=black!20] (-1.7,-1) -- (-2.5,-1) -- (-2.1,-2) -- cycle;
\node at (-2.1,-1.4) {$\chi_x$} ;
\draw[black,fill=black!20] (-1.7,1) -- (-2.5,1) -- (-2.1,2) -- cycle;
\node at (-2.1,1.4) {$\chi_x$} ;
\draw[black,fill=black!20] (4.5,1) -- (5.3,1) -- (4.9,2) -- cycle;
\node at (4.9,1.4) {$\zeta_z$} ;
\draw[black,fill=black!20] (4.5,-1) -- (5.3,-1) -- (4.9,-2) -- cycle;
\node at (4.9,-1.4) {$\zeta_z$} ;
\draw (4.9,-1) -- (4.9,-.5) ;
\draw (4.9,1) -- (4.9,.5) ;
\draw (-2.1,-1) -- (-2.1,-.5) ;
\draw (-2.1,1) -- (-2.1,.5) ;
\draw (.6,1) -- (.6,.5) ;
\draw (.6,-1) -- (.6,-.5) ;
\draw (-.6,1) -- (-.6,.5) ;
\draw (-.6,-1) -- (-.6,-.5) ;
\draw (3.4,1) -- (3.4,.5) ;
\draw (3.4,-1) -- (3.4,-.5) ;
\draw (2.2,1) -- (2.2,.5) ;
\draw (2.2,-1) -- (2.2,-.5) ;
\draw[red,dashed] plot [smooth cycle,tension=.3] coordinates { (-1.4,-0.8) (-1.6,-2.2) (-2.7,-2.1) (-2.7,2.1) (-1.6,2.2) (-1.4,0.8) (-0.2,0.7) (-0.2,-0.7) } ;
\draw[red,dashed] plot [smooth cycle,tension=.3] coordinates { (4.2,-0.8) (4.4,-2.2) (5.5,-2.1) (5.5,2.1) (4.4,2.2) (4.2,0.8) (3.0,0.7) (3.0,-0.7) } ;
\end{tikzpicture}
$$
Here, the dashed line indicates how to consider $A^x_a\equiv\langle\chi_x|A_a|\chi_x\rangle$, respectively $C^z_c\equiv\langle\zeta_z|C_c|\zeta_z\rangle$, as operators acting on one part of the bipartite state $|\psi_{\T{AB}}\rangle$, respectively $|\psi_{\T{BC}}\rangle$. By $\sum_a A_a=\mathbbm{1}$ and normalization of $|\chi_x\rangle$, it follows that $\sum_a A^x_a=\mathbbm{1}$ for all $x$; similarly, $\sum_z C^z_c=\mathbbm{1}$ for all $z$. By definition,~(\ref{P4quantum}) can then be written as
\beq
\label{BGPquantum}
p(a,b,c|x,z) = \left(\langle\psi_{\T{AB}}| \otimes \langle\psi_{\T{BC}}|\right)\left(A^x_a\otimes B_b\otimes C^z_c\right)\left(|\psi_{\T{AB}}\rangle\otimes|\psi_{\T{BC}}\rangle\right) .
\eeq
This is desired quantum representation of $p$ in a BGP scenario.

Conversely, we start from a correlation $p(a,b,c,x,z)$ of the form~(\ref{BGPquantum}). As sources between \T{A} and \T{X} and between \T{C} and \T{Z}, we again take hidden variables defined by $\lambda_{\T{X}\T{A}}=x$ and $\lambda_{\T{C}\T{Z}}=z$; again, the protocol of \T{X} and \T{Z} is simply to announce the values of these variables as their outcome. Only the sources between \T{A} and \T{B} and between \T{B} and \T{C} are taken to be quantum and produce, respectively, the bipartite states $|\psi_{\T{AB}}\rangle$ and $|\psi_{\T{BC}}\rangle$ of~(\ref{BGPquantum}). The measurement protocol conducted by \T{A} is similar to above: measure $\lambda_{\T{X}\T{A}}$, use the result as the choice of setting for the subsequent measurement on $|\psi_{\T{AB}}\rangle$, and then announce both outcomes as the total outcome. This protocol can be interpreted as measuring a single POVM given by
$$
|x\rangle\langle x|\otimes A^x_a \mapsto (x,a),
$$
where the left-hand side is a POVM element indexed by $x$ and $a$, and the right-hand side denotes the resulting outcome announced by \T{A}. The analogous POVM is measured by \T{C}. By construction, this reproduces both the desired conditional distribution~(\ref{BGPquantum}) and the marginal distribution $p(x,z)=p(x)p(z)$, and therefore also the whole distribution $p(x,a,b,c,z)$.
\end{asparaenum}
\end{proof}

Due to this theorem, we can regard $P_5$ as the analogue of the BGP scenario within our formalism.

However, this is not yet the end of the story; our new point of view provides more than just a reformulation of familiar things. Let us imagine that party \T{Z}, in the $P_5$ scenario, has failed to collect data. Or that we disregard \T{Z}'s measurement for some other reason. Then, we can regard the remaining parties \T{X}, \T{A}, \T{B}, \T{C} as forming a $P_4$ scenario and apply Theorem~\ref{P4thm} to the distribution $p(x,a,b,c)$, with $c$ now playing the role of $y$. In this way, the $P_4$ scenario is a natural \emph{subscenario} of $P_5$. This is an observation which does not make sense in the standard formalism.

\subsection*{The triangle scenario $C_3$}
\label{triangle}

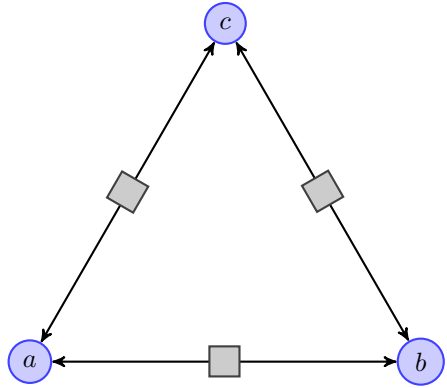
\begin{figure}
\begin{tikzpicture}[node distance=2.1cm,>=stealth',thick]
\tikzstyle{place}=[circle,thick,draw=blue!75,fill=blue!20,minimum size=4mm]
\tikzstyle{transition}=[rectangle,thick,draw=black!75,fill=black!20,minimum size=4.0mm]
\node[place] (a) at (210:3) {$a$} ;
\node[place] (b) at (330:3) {$b$} ;
\node[place] (c) at (90:3) {$c$} ;
\draw[<->] (b) -- (c) node [midway,sloped,above=-6pt,transition] {} ;
\draw[<->] (a) -- (b) node [midway,sloped,above=-6pt,transition] {} ;
\draw[<->] (a) -- (c) node [midway,sloped,above=-6pt,transition] {} ;
\end{tikzpicture}
\caption{The correlation scenario $C_3$.}
\label{trianglefig}
\end{figure}

Our next example, first proposed in~\cite{BRGP}*{Sec.VI}, is the correlation scenario illustrated in Figure~\ref{trianglefig}. It consists of three parties of which each two share a common source. We will see in Corollary~\ref{C3simplest} that it is the smallest scenario in which non-classical correlations exist. In this subsection, we prove the existence of non-classical quantum correlations in $C_3$.

We find this scenario especially appealing both due to its symmetry and due to its appearance in the study of inference of common ancestors~\cite{SA}; see below. Since the main ideas concerning correlation scenarios should already have become clear in the last two examples, we now increase the pace a bit.

\begin{defn}
\label{C3corrdef}
A correlation in $C_3$ is a distribution $p(a,b,c)$. (It is not required to satisfy any particular constraint.)
\end{defn}

This definition seems reasonable to us since, in general, one cannot expect any two of the variables $(a,b,c)$ to be independent.

\begin{expl}
\label{pc}
If all three variables take values in $\{0,1\}$, then 
$$
p(a=b=c=0)=\tfrac{1}{2},\qquad p(a=b=c=1)=\tfrac{1}{2}
$$
defines a correlation. We call this the \emph{perfect correlation} since all three variables are random, but perfectly correlated.
\end{expl}

\begin{defn}
\label{C3classdef}
A correlation $p(a,b,c)$ is classical in the $C_3$ scenario if and only if it can be written in the form
\beq
\label{C3lhv}
p(a,b,c) = \int_{\lambda_{\T{A}\T{B}},\lambda_{\T{B}\T{C}},\lambda_{\T{C}\T{A}}} p(a|\lambda_{\T{C}\T{A}},\lambda_{\T{A}\T{B}}) p(b|\lambda_{\T{A}\T{B}},\lambda_{\T{B}\T{C}}) p(c|\lambda_{\T{B}\T{C}},\lambda_{\T{C}\T{A}}) p(\lambda_{\T{A}\T{B}}) p(\lambda_{\T{B}\T{C}}) p(\lambda_{\T{C}\T{A}})
\eeq
for appropriate (conditional) distributions $p(a|\lambda_{\T{C}\T{A}}$, $\lambda_{\T{A}\T{B}})$, $p(b|\lambda_{\T{A}\T{B}}$, $\lambda_{\T{B}\T{C}})$, $p(c|\lambda_{\T{B}\T{C}},\lambda_{\T{C}\T{A}})$, $p(\lambda_{\T{A}\T{B}})$, $p(\lambda_{\T{B}\T{C}})$, $p(\lambda_{\T{C}\T{A}})$.
\end{defn}

Classical correlations in $C_3$ are monogamous in the following sense:

\begin{prop}
\label{monogamy}
Let $p(a,b,c)$ be classical. If $p(a=c)=1$, then  $a$ is independent of $\lambda_{\T{A}\T{B}}$.
\end{prop}

Intuitively, this is because in order to create these perfect correlations between $a$ and $c$, the outcome $a$ cannot depend on $\lambda_{\T{A}\T{B}}$. In particular, this implies that there cannot be any correlations between $a$ and $b$. Rigorously, the proof technique is the same as the one used in the proof of this inequality relating Shannon entropy and mutual information, which can be regarded as a monogamy inequality:

\begin{lem}
\label{C3ei}
Let $p(a,b,c)$ be classical. Then
\beq
\label{ei}
I(a:b) + I(a:c) \leq H(a) .
\eeq
\end{lem}

The interpretation of this is a kind of monogamy: $a$ can share strong correlations with only $b$ or $c$, but not with both. In particular, this inequality shows that the perfect correlation of Example~\ref{pc} is not classical.

\begin{proof}[Proof of Proposition~\ref{monogamy} and Lemma~\ref{C3ei}]
The present proof concerns the case that the hidden variables are discrete; see~\ref{eigenproof} for the general case.

Since $a$ and $b$ are conditionally independent given $\lambda_{\T{A}\T{B}}$, and similarly for $a$ and $c$, the data processing inequality can be used to bound the left-hand side of\eq{ei} by
$$
I(a:b) + I(a:c) \leq I(a:\lambda_{\T{A}\T{B}}) + I(a:\lambda_{\T{C}\T{A}}) = 2H(a) + H(\lambda_{\T{A}\T{B}}) + H(\lambda_{\T{C}\T{A}}) - H(a\lambda_{\T{A}\T{B}}) - H(a\lambda_{\T{C}\T{A}}).
$$
Submodularity of Shannon entropy guarantees that $H(a\lambda_{\T{A}\T{B}}) + H(a\lambda_{\T{C}\T{A}}) \geq H(a) + H(a\lambda_{\T{A}\T{B}}\lambda_{\T{C}\T{A}})$, which can be applied here to obtain
$$
I(a:b) + I(a:c) \leq H(a) + H(\lambda_{\T{A}\T{B}}) + H(\lambda_{\T{C}\T{A}}) - H(a\lambda_{\T{A}\T{B}}\lambda_{\T{C}\T{A}}) \leq H(a) + I(\lambda_{\T{A}\T{B}}:\lambda_{\T{C}\T{A}}) .
$$
Since $I(\lambda_{\T{A}\T{B}}:\lambda_{\T{C}\T{A}})=0$, the claim of Lemma~\ref{C3ei} follows.

Concerning Proposition~\ref{monogamy}, its assumption implies $I(a:c)=H(a)$; the sequence of inequalities derived in this proof then guarantees that $I(a:\lambda_{\T{A}\T{B}})=0$, as was to be shown.
\end{proof}

\begin{cor}
\label{moncor}
Let $p(a,b,c)$ be classical and $f,g$ functions such that $f(a)$ and $g(c)$ are defined. If $p\left(f(a)=g(c)\right)=1$, then $f(a)$ and $g(c)$ are independent of $\lambda_{\T{A}\T{B}}$.
\end{cor}

\begin{proof}
The assumptions imply that $p(f(a),b,g(c))$ is also a classical correlation in $C_3$. Now the claim follows from Proposition~\ref{monogamy}.
\end{proof}

\begin{thm}
\label{C3thm}
There exist non-classical quantum correlations in $C_3$.
\end{thm}

\begin{proof}
We take $|\psi\rangle$ to be a bipartite two-qubit state which violates the CHSH inequality~\cite{CHSH} with respect to measurements in the two bases $\{|\phi_0,\rangle,|\phi_1\rangle\}$, $\{|\omega_0,\rangle,|\omega_1\rangle\}$, which are the same for both parties.

The quantum correlations we consider in $C_3$ are obtained as follows. We take \T{A} and \T{B} to share $|\psi\rangle$, while \T{A} and \T{C} as well as \T{B} and \T{C} share either a maximally entangled state
\beq
\label{entstate}
\frac{|00\rangle + |11\rangle}{\sqrt{2}}
\eeq
or, equivalently, a classically correlated mixed state
\beq
\label{mixstate}
\frac{1}{2}\left(|00\rangle\langle 00| + |11\rangle\langle 11|\right)
\eeq
of two qubits. The purpose of these states is simple: it obsoletes free will in that \T{A} and \T{B} first measure the system they receive from the source shared with \T{C} in the $\{|0\rangle,|1\rangle\}$-basis and use the resulting outcome as a measurement setting on $|\psi\rangle$; this is similar to how the proofs of Theorems~\ref{P4thm} and~\ref{P5thm} work. \T{A} and \T{B} announce the outcomes of both measurements as their total outcome. Similarly, we take \T{C} to apply the $\{|0\rangle,|1\rangle\}$-measurement on each of his qubits, so that \T{C} knows the measurement ``setting'' used by \T{A} and \T{B}. He announces both of them as his outcome $c$. We regard the two bits announced by each party as the outcome of a single four-outcome measurement. The resulting correlation $p(a,b,c)$ is a probability distribution on $4^3$ outcomes which does not depend on whether\eq{entstate} or\eq{mixstate} is used.

More formally, we can define the measurements as follows: both \T{A} and \T{B} measure in the following basis and announce respective outcomes:
\begin{align*}
|0\rangle\langle 0| \otimes |\phi_0\rangle\langle\phi_0| \mapsto (0,0),\quad 
|0\rangle\langle 0| \otimes |\phi_1\rangle\langle\phi_1| \mapsto (0,1), \\
|1\rangle\langle 1| \otimes |\omega_0\rangle\langle\omega_0| \mapsto (1,0),\quad 
|1\rangle\langle 1| \otimes |\omega_1\rangle\langle\omega_1| \mapsto (1,1),
\end{align*}
while \T{C} simply measures both his qubits in the standard basis and announces both results. 

It needs to be proven that these correlations are non-classical in $C_3$. This is guaranteed by the monogamy property of Corollary~\ref{moncor}: since \T{C} has perfect information about the ``settings'' employed by \T{A} and \T{B}, these ``settings'' are necessarily indepedendent of $\lambda_{\T{A}\T{B}}$. This simulates the free will (``$\lambda$-independence'') required for a standard Bell test to apply. The hidden variable $\lambda_{\T{A}\T{B}}$ in any potential classical model would therefore have to function exactly like a hidden variable in a standard Bell scenario, which is guaranteed to be impossible due to the Bell inequality violation.
\end{proof}

These arguments apply in the same way to a construction of a non-classical quantum correlations from a Bell inequality violation in any bipartite Bell scenario.

Although this class of examples proves the theorem, we do not find such examples satisfying since they are again based on a Bell test in the standard sense. It is difficult to regard them as entirely new kinds of non-classicality. Nevertheless, we find it surprising that non-classical quantum correlations exist in $C_3$ even in the case when \emph{only one} of the sources produces entanglement. We had not expected this at all when we started thinking about the $C_3$ scenario.

\begin{prob}
Find an example of non-classical quantum correlations in $C_3$ together with a proof of its non-classicality which does not hinge on Bell's Theorem.
\end{prob}

In order to find more examples of non-classical quantum correlations in $C_3$, it would be helpful to have inequalities bounding the set of classical correlations and violated by some quantum correlations. Unfortunately, our proof of Theorem~\ref{C3thm} does provide inequalities only conditional on the perfect correlations required between \T{A} and \T{C} and between \T{B} and \T{C}. However, we expect that our idea can be used to derive unconditional inequalities, if one knows bounds on the maximal classical value of a Bell inequality as a function of the correlation between the measurement settings and the hidden variable. We expect that such bounds can be derived by considerations similar to those of~\cite{BG} and/or~\cite{CR} or may even be implicitly contained in these works.

Before moving on to the next example of a correlation scenario, we return briefly to the work of Steudel and Ay~\cite{SA} on the inference of common ancestors. So, what is a ``common ancestor''?

If one makes certain (say, real-world macroscopic) observations $a$ and $b$, repeats them many times in order to gather statistics, and detects a correlation between these, then one can conclude that $a$ and $b$ need to have a \emph{common ancestor}: there needs to be some quantity or property $\lambda$ such that both $a$ and $b$ depend on $\lambda$, and $\lambda$ is not deterministic; this includes the possibilities $\lambda=a$ and $\lambda=b$ as degenerate cases. This $\lambda$ is a common ancestor of $a$ and $b$ in the sense of a preexisting condition on which both $a$ and $b$ depend.

This is Reichenbach's principle of common cause~\cites{Reichenbach,Eber}; it is based on the premise that good models of the world adhere to assumption~\ref{IS} in the sense that a good model should predict $a$ and $b$ to be independent, unless there is some previously occurring event causally connected to both variables, i.e.~a common ancestor.

Now what if one does the same for three observations $a$, $b$, $c$? How can one conclude that there is a common ancestor $\lambda$ on which all three of them depend? Or for any number $n\in \N$ of observations? Among other things, it has been shown in~\cite{SA} that the entropy of the common ancestors is lower bounded by a certain linear combination of the joint entropy and the single-variable entropies; therefore, strict positivity of that linear combination witnesses the necessity of a common ancestor. See also~\cite{Ay2} for related work providing a generalization and quantification of Reichenbach's principle.

Let us consider the particular case of $n=3$ variables. Then the main observation is that the causal structure of the $C_3$ scenario is precisely the null hypothesis: if no common ancestor exists, then there can at most be common ancestors for every pair of variables, but not for all three variables together. Therefore, \emph{if no common ancestor exists, then $p(a,b,c)$ is a classical correlation in the $C_3$ scenario.} Figure~\ref{trianglefig} coincides with~\cite{SA}*{Figure~1}. The results of Steudel and Ay for this particular case state that if $p(a,b,c)$ is classical, then
\beq
\label{SAei}
H(a) + H(b) + H(c) \leq 2 H(abc),
\eeq
where $H(abc)$ is the entropy of the joint distribution. Intuitively, if this inequality is violated, then the joint entropy is relatively small in comparison to the single-variable entropies, implying the existence of strong correlations between the variables and therefore of a common ancestor. 

Writing out our inequality~(\ref{ei}) in terms of joint entropies, one obtains
$$
H(a) + H(b) + H(c) \leq H(ab) + H(ac),
$$
which is an improvement over~(\ref{SAei}) since the right-hand side is bounded by $2H(abc)$. In particular, a violation
\beq
\label{newei}
H(a) + H(b) + H(c) > H(ab) + H(ac)
\eeq
successfully witnesses the necessity of a common ancestor in strictly more cases than\eq{SAeiI}.

In the case of $n>3$ variables, it is still true that the null hypothesis of non-existence of a common cause corresponds to classicality in the appropriate correlation scenario: for the necessity of a common ancestor of some $(k+1)$-element subset of $n$ variables, the null hypothesis is that at most each $k$-tuple has common ancestor(s). Roughly speaking, it is enough to consider only those ancestors which themselves do not have any parents: all the randomness creation can be delegated to those without changing the distribution of the observed variables, while all other nodes then carry out deterministic information processing; compare Remark~\ref{wlogdet} and~\ref{appwlogdet}. Then each such initial node can be replaced by a source connecting to at most $k$ observed variables, and the deterministic information processing can as well delegated to the measurement nodes, again without changing the distribution of outcomes. Therefore, this corresponds to a classical correlation in the correlation scenario defined by $n$ measurements in which each $k$-tuple of measurements is allowed to share a source. Conversely, it is clear that every such classical correlation represents a joint distribution of $n$ variables which can be modelled without a common ancestor for any $(k+1)$-tuple. To summarize, the given joint distribution is a classical correlation in this scenario if and only if the joint distribution can be obtained from a Bayesian network in which no $(k+1)$-element subset of the given variables has a common ancestor.

However, at the moment we do not know how to generalize our inequality\eq{newei} to these cases, and refer once again to~\cite{SA} for the current state of the art.

\subsection*{The square scenario $C_4$.}

\begin{figure}
\centering
\begin{tikzpicture}[node distance=2.3cm,>=stealth',thick]
\tikzstyle{place}=[circle,thick,draw=blue!75,fill=blue!20,minimum size=4mm]
\tikzstyle{transition}=[rectangle,thick,draw=black!75,fill=black!20,minimum size=4.0mm]
\node[place] (a) at (0,0) {$a$} ;
\node[place] (x) [above=of a] {$x$} ;
\node[place] (b) [right=of a] {$b$} ;
\node[place] (y) [above=of b] {$y$} ;
\draw[<->] (x) -- (a) node [midway,above=-6pt,transition] {} ;
\draw[<->] (a) -- (b) node [midway,above=-6pt,transition] {} ;
\draw[<->] (y) -- (b) node [midway,above=-6pt,transition] {} ;
\draw[<->] (x) -- (y) node [midway,above=-6pt,transition] {} ;
\end{tikzpicture}
\caption{The correlation scenario $C_4$.}
\label{C4fig}
\end{figure}
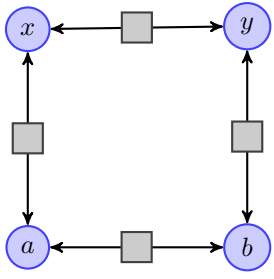

Another interesting correlation scenario is the \emph{square scenario} illustrated in Figure~\ref{C4fig}. In this case, the underlying graph is $C_4$, the cycle graph on four vertices. It can be regarded as $P_4$ (Figure~\ref{P4fig}) equipped with an additional source between \T{X} and \T{Y}. Along the lines of Theorem~\ref{P4thm}, this would suggest that correlations $p(a,b,x,y)$ in $C_4$ should be interpretable as arising from a Bell scenario together with correlations between the measurement settings. However, the forthcoming Proposition~\ref{C4prop} will show that this intuition is false.

\begin{defn}
\label{C4corrdef}
A \emph{correlation} $p$ in the $C_4$ scenario is a distribution $p(a,b,x,y)$ whose marginals factorize as
$$
p(a,y) = p(a)p(y),\qquad p(b,x) = p(b) p(x).
$$
\end{defn}

\begin{defn}
\label{C4classdef}
A correlation $p(a,b,x,y)$ is classical in the $C_4$ scenario if and only if it can be written in the form
\beq
\label{C4lhv}
p(a,b,x,y) = \int_{\lambda_{\T{A}\T{B}},\lambda_{\T{B}\T{Y}},\lambda_{\T{Y}\T{X}},\lambda_{\T{X}\T{A}}} p(a|\lambda_{\T{X}\T{A}},\lambda_{\T{A}\T{B}}) p(b|\lambda_{\T{A}\T{B}},\lambda_{\T{B}\T{Y}}) p(y|\lambda_{\T{B}\T{Y}},\lambda_{\T{Y}\T{X}}) p(x|\lambda_{\T{Y}\T{X}},\lambda_{\T{X}\T{A}}) p(\lambda_{\T{A}\T{B}}) p(\lambda_{\T{B}\T{Y}}) p(\lambda_{\T{Y}\T{X}}) p(\lambda_{\T{X}\T{A}})
\eeq
for appropriate (conditional) distributions $p(a|\lambda_{\T{X}\T{A}},\lambda_{\T{A}\T{B}})$, $p(b|\lambda_{\T{A}\T{B}},\lambda_{\T{B}\T{Y}})$, $p(y|\lambda_{\T{B}\T{Y}},\lambda_{\T{Y}\T{X}})$, $p(x|\lambda_{\T{Y}\T{X}},\lambda_{\T{X}\T{A}})$, $p(\lambda_{\T{A}\T{B}})$, $p(\lambda_{\T{B}\T{Y}})$, $p(\lambda_{\T{Y}\T{X}})$, $p(\lambda_{\T{X}\T{A}})$.
\end{defn}

\begin{prop}
\label{C4prop}
There are classical correlations $p(a,b,x,y)$ in the $C_4$ scenario such that the associated conditional distribution $p(a,b|x,y)$ is signaling.
\end{prop}

\begin{proof}
We start from any classical correlation $p_0(x,a,b,y)$ in the $P_4$ scenario. In particular, by Theorem~\ref{P4thm}, $p_0(a,b|x,y)$ does not violate any Bell inequality. We now apply the relabeling
$$
a\longleftrightarrow x,\qquad b\longleftrightarrow y
$$
and take the resulting correlation to be $p(a,b,x,y)$. By construction, the resulting correlation $p(a,b,x,y)$ is classical in $C_4$. By construction, $p(x,y|a,b)$ does not violate a Bell inequality. The conditional distribution
$$
p(a,b|x,y) = p(x,y|a,b) \cdot\frac{p(x,y)}{p(a)p(b)}
$$
then is precisely the time reversal, in the sense of Coecke and Lal~\cite{CL}, of the classical no-signaling box $p(x,y|a,b)$ with respect to $p(a,b)=p(a)p(b)$ as its distribution of settings. It was shown in~\cite{CL} that there exist $p_0(a,b|x,y)$ for which this time reversal is necessarily signaling.
\end{proof}

In particular, Proposition~\ref{C4prop} shows that the conditional distribution $p(a,b|x,y)$ associated to a classical correlation $p(a,b,x,y)$ in $C_4$ may \emph{violate Bell inequalities}.

Any classical (resp.~quantum) correlation in a bipartite Bell scenario can be turned into a classical (resp.~quantum) correlation in the $C_4$ scenario in four different ways: one of the four edges of $C_4$ needs to be designated as the Bell scenario's source, while the source corresponding to the opposite edge does nothing at all.

\begin{thm}
\label{C4thm}
There exist non-classical correlations in $C_4$.
\end{thm}

\begin{proof}
We define the correlation $p(a,b,x,y)$ by taking $p(a,b|x,y)$ to be a Popescu-Rohrlich box~\cite{PR} and $p(x,y)$ to be the uniform distribution. More concretely, all four outcomes are bits $a,b,x,y\in\{0,1\}$ with the table of joint probabilities given by:
\beq
\label{PRbox}
\textrm{
\begin{tabular}{cc|cccc}
&& \multicolumn{4}{c}{$(x,a)=$} \\
&& $(0,0)$ & $(0,1)$ & $(1,0)$ & $(1,1)$ \\
\hline
\multirow{4}{*}{$(y,b)=$} & $(0,0)$ & $\tfrac{1}{8}$ & $0$ & $\tfrac{1}{8}$ & $0$ \\
& $(0,1)$ & $0$ & $\tfrac{1}{8}$ & $0$ & $\tfrac{1}{8}$ \\
& $(1,0)$ & $\tfrac{1}{8}$ & $0$ & $0$ & $\tfrac{1}{8}$ \\
& $(1,1)$ & $0$ & $\tfrac{1}{8}$ & $\tfrac{1}{8}$ & $0$ 
\end{tabular}}
\eeq
We now use a Hardy-type~\cite{Hardy} argument in order to show that this correlation is not classical. For the sake of contradiction, let us assume $p(a,b,x,y)$ to be classical with hidden variable distributions $p(\lambda_{\T{A}\T{B}})$, $p(\lambda_{\T{B}\T{Y}})$, $p(\lambda_{\T{Y}\T{X}})$, $p(\lambda_{\T{X}\T{A}})$; thanks to Remark~\ref{wlogdet}, we can take the four outcomes to be deterministic funtions of the hidden variables. We start by considering the case of discrete hidden variables. Then, there has to be a hidden variable combination
$$
(\lambda_{\T{A}\T{B}},\lambda_{\T{B}\T{Y}},\lambda_{\T{Y}\T{X}},\lambda_{\T{X}\T{A}}) = (\ell_{\T{A}\T{B}},\ell_{\T{B}\T{Y}},\ell_{\T{Y}\T{X}},\ell_{\T{X}\T{A}})
$$
occuring with positive probability, which produces the outcome $(a,b,x,y)=(0,0,0,0)$; similarly, there has to be a hidden variable combination
$$
(\lambda_{\T{A}\T{B}},\lambda_{\T{B}\T{Y}},\lambda_{\T{Y}\T{X}},\lambda_{\T{X}\T{A}}) = (\kappa_{\T{A}\T{B}},\kappa_{\T{B}\T{Y}},\kappa_{\T{Y}\T{X}},\kappa_{\T{X}\T{A}}),
$$
occuring with positive probability, which produces the outcome $(a,b,x,y)=(1,0,1,1)$. Then, the independence of sources guarantees that the hidden variable combination
$$
(\lambda_{\T{A}\T{B}},\lambda_{\T{B}\T{Y}},\lambda_{\T{Y}\T{X}},\lambda_{\T{X}\T{A}}) = (\kappa_{\T{A}\T{B}},\kappa_{\T{B}\T{Y}},\ell_{\T{Y}\T{X}},\kappa_{\T{X}\T{A}}),
$$
also has positive probability. Because of locality and determinism, it necessarily produces an outcome $(1,0,\hat{x},\hat{y})$; by\eq{PRbox}, $\hat{x}=\hat{y}=1$. Likewise, the hidden variable combination $(\ell_{\T{A}\T{B}},\ell_{\T{B}\T{Y}},\ell_{\T{Y}\T{X}},\kappa_{\T{X}\T{A}})$ has positive probality, and produces some outcome of the form $(a',0,1,0)$. Thanks to the form of\eq{PRbox}, necessarily $a'=0$. Similarly, the hidden variable combination $(\ell_{\T{A}\T{B}},\kappa_{\T{B}\T{Y}},\ell_{\T{Y}\T{X}},\ell_{\T{X}\T{A}})$ must give the outcome $(0,0,0,1)$. However, the hidden variable combination $(\ell_{\T{A}\T{B}},\kappa_{\T{B}\T{Y}},\ell_{\T{Y}\T{X}},\kappa_{\T{X}\T{A}})$ then gives the outcome $(0,0,1,1)$ with positive probability, a contradiction with\eq{PRbox}.

In the case of general (non-discrete) hidden variables, the same proof idea can be used, although the technical details are quite involved; see~\ref{C4proof}.
\end{proof}

\begin{prob}
\label{C4prob}
\begin{enumerate}
\item\label{C4qstnq} Are there non-classical quantum correlations in $C_4$?
\item Is there a simple way to characterize the classical correlations in $C_4$?
\end{enumerate}
\end{prob}

\subsection*{Scenarios with multipartite sources}

So far, we have only considered example scenarios in which each source produces a pair of systems which it distributes among two parties. However, it is quite common to consider Bell scenarios involving a source that distributes systems among several parties~\cite{GHZ}. The same can be easily done in our framework; an example scenario of this type is illustrated in Figure~\ref{Skfig}.

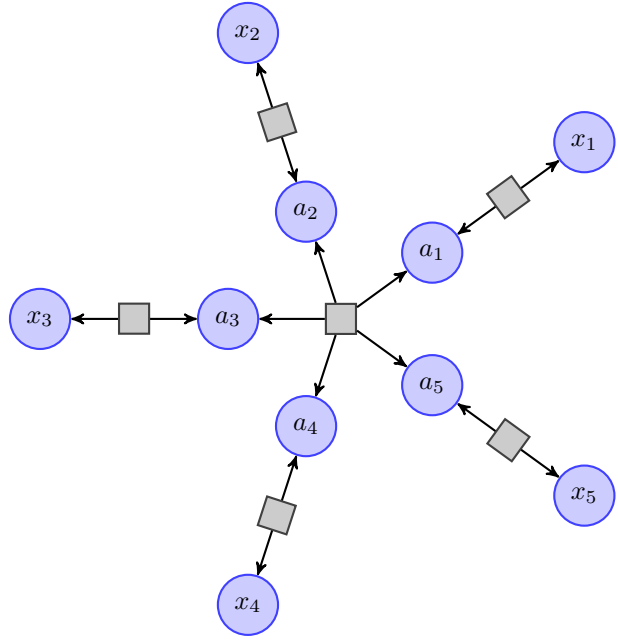
\begin{figure}
\begin{tikzpicture}[node distance=2.1cm,>=stealth',thick]
\tikzstyle{place}=[circle,thick,draw=blue!75,fill=blue!20,minimum size=8mm]
\tikzstyle{transition}=[rectangle,thick,draw=black!75,fill=black!20,minimum size=4.0mm]
\node[place] (a1) at (36:1.5) {$a_1$} ;
\node[place] (a2) at (108:1.5) {$a_2$} ;
\node[place] (a3) at (180:1.5) {$a_3$} ;
\node[place] (a4) at (252:1.5) {$a_4$} ;
\node[place] (a5) at (324:1.5) {$a_5$} ;
\node[place] (x1) at (36:4) {$x_1$} ;
\node[place] (x2) at (108:4) {$x_2$} ;
\node[place] (x3) at (180:4) {$x_3$} ;
\node[place] (x4) at (252:4) {$x_4$} ;
\node[place] (x5) at (324:4) {$x_5$} ;
\node[transition] (l) at (0,0) {} ;
\draw[->] (l) -- (a1) ;
\draw[->] (l) -- (a2) ;
\draw[->] (l) -- (a3) ;
\draw[->] (l) -- (a4) ;
\draw[->] (l) -- (a5) ;
\draw[<->] (a1) -- (x1) node [midway,sloped,above=-6pt,transition] {} ;
\draw[<->] (a2) -- (x2) node [midway,sloped,above=-6pt,transition] {} ;
\draw[<->] (a3) -- (x3) node [midway,sloped,above=-6pt,transition] {} ;
\draw[<->] (a4) -- (x4) node [midway,sloped,above=-6pt,transition] {} ;
\draw[<->] (a5) -- (x5) node [midway,sloped,above=-6pt,transition] {} ;
\end{tikzpicture}
\caption{The correlation scenario $A_5$.}
\label{Skfig}
\end{figure}

More generally, we want to consider the family of \emph{multiarm scenarios} $A_k$ indexed by the number of arms $k\in\N$; each arm consists of two parties sharing a bipartite source, and there is one $k$-partite source shared by all the parties obtained by choosing one party in each arm. Figure~\ref{Skfig} represents the case $k=5$, while $k=2$ is the $P_4$ scenario of Figure~\ref{P4fig}.

The following considerations are immediate generalizations of those of the $P_4$ scenario. Just as $P_4$ corresponds to a bipartite Bell scenario, $A_k$ corresponds to a $k$-partite Bell scenario. We use ``hat'' notation like $a_1,\ldots,\hat{a}_i,\ldots,a_k$ as short for $a_1,\ldots,a_{i-1},a_{i+1},\ldots,a_k$.

\begin{defn}
\label{Skcorrdef}
A \emph{correlation} in the $A_k$ scenario is a probability distribution $p(a_1,\ldots,a_k,x_1\ldots,x_k)$ whose marginals factorize as
\beq
\label{Sknosig}
p(a_1,\ldots,\hat{a}_i,\ldots,a_k,x_1,\ldots,x_k) = p(x_i) p(a_1,\ldots,\hat{a}_i,\ldots,a_k,x_1,\ldots,\hat{x}_i,\ldots,x_k). \quad \forall i
\eeq
\end{defn}

Repeated application of\eq{Sknosig} implies $p(x_1,\ldots,x_n) = p(x_1) \cdots p(x_n)$. Upon using this, and considering only those values $x_i$ for which $p(x_i)>0$, the condition\eq{Sknosig} becomes equivalent to the equations
$$
p(a_1,\ldots,\hat{a}_i,\ldots,a_n|x_1,\ldots,x_n) = p(a_1,\ldots,\hat{a}_i,\ldots,a_n|x_1,\ldots,\hat{x}_i,\ldots,x_n) ,
$$
which are formally identical to the no-signaling equations in a $k$-partite Bell scenario.

\begin{thm}
\label{Skthm}
\begin{enumerate}
\item\label{Skcl} A correlation $p$ is \emph{classical} in $A_k$ if and only if the associated conditional distribution $p(a_1,\ldots,a_k|x_1,\ldots,x_k)$ is classical in the Bell scenario sense.
\item\label{Skqu} A correlation $p(a,b,x,y)$ is \emph{quantum} in $A_k$ if and only if the associated conditional distribution $p(a_1,\ldots,a_k|x_1,\ldots,x_k)$ is quantum in the Bell scenario sense
\end{enumerate}
\end{thm}

\begin{proof}
Analogous to the proof of Theorem~\ref{P4thm}.
\end{proof}

\section{General theory of correlation scenarios}
\label{gentlcs}

We now adopt a more abstract point of view. Looking at the previous examples, one should come to the conclusion that a general definition of correlation scenario should define the data of a correlation scenario to consist of a set of measurements ($=$ parties $=$ observers) $M$, a set of sources $S$, and a relation $C\subseteq S\times M$ between sources and measurements, where we write $(s,m)\in C$ also as $sCm$ and read it as ``$s$ connects to $m$''. As before, the physical picture is that each source sends out one physical system to each party it connects to, and each party conducts a fixed measurement on the collection of systems it receives from the sources it is connected to. The temporal (or rather causal) structure of such a scenario consists of a primary layer of sources and a secondary layer of measurements. In~\cite{FS}, we will go beyond this ``two-layer'' approach and consider a vastly more general formalism allowing for any kind of causal structure.

Finally, a correlation scenario should also specify how many possible outcomes each measurement has. For simplicity, we take this to be the same number $d\in\N$ for all measurements. We usually omit mention of $d$ and regard it as implicitly defined through the correlation: given the joint outcome distribution, $d$ can be taken to be equal to the highest number of actually occurring outcomes over all measurements.

\begin{defn}
\label{2layerdef}
A \emph{correlation scenario} is a quadruple $(S,M,C,d)$ consisting of a finite set of \emph{sources} $S$, a finite set of \emph{measurements} $M$, a relation $C\subseteq S\times M$ (read: ``connects'') and a natural number $d\in\N$. The relation is required to satisfy the conditions
\begin{enumerate}
\item $\left( s_1Cm\Rightarrow s_2Cm \:\forall m \right) \:\Longleftrightarrow\:  s_1 = s_2$
\item $\left( sCm_1\Leftrightarrow sCm_2 \:\forall s \right) \:\Longleftrightarrow\:  m_1 = m_2$
\end{enumerate}
\end{defn}

These two conditions are to be interpreted as follows: if source $s_2$ connects to each measurement to which also $s_1$ connects, then $s_1$ is redundant. Therefore, we may assume without loss of generality that such redundancies do not occur: if $s_1$ connects to a subset of the measurements to which $s_2$ connects, or to exactly the same measurements, then $s_1=s_2$. Similarly, if there are two measurements which connect to exactly the same set of sources, then we may replace both measurements by a single one. Therefore, we assume without loss of generality that if $m_1$ and $m_2$ connect to the same set of sources, then $m_1=m_2$.

The scenarios depicted in Figures~\ref{P4fig}--\ref{Skfig} are exactly of this form: the circles represent $M$, the boxes form $S$, and the arrows define $C$.

\begin{defn}
A \emph{hypergraph} $G=(V,E)$ consists of a finite set of vertices $V$ and a set of edges $E\subseteq 2^V$, i.e.~every edge $e\in E$ is a subset $e\subseteq V$.
\end{defn}

The combinatorial data of Definition~\ref{2layerdef} can equivalently be specified in terms of a hypergraph. One obtains a hypergraph from a correlation scenario $(S,M,C,d)$ by using the vertex set $V=M$ and introducing one edge for each source which contains exactly those vertices ($=$ measurements) to which the source connects. Formally, the resulting set of edges is
$$
E=\left\{ \{r\in P\::\: sCr\} ,\: s\in S\right\} .
$$
Then the two requirements of Definition~\ref{2layerdef} translates into the properties
\begin{enumerate}
\item\label{hgp1} $G$ is an anti-chain: there is no edge which is contained in a different one.
\item\label{hgp2} There are no two different vertices which belong to exactly the same set of edges.
\end{enumerate}
Conversely, every hypergraph with these properties defines a correlation scenario in the obvious way: vertices become measurements, and every edge defines a source which connects to all those measurements contained in the edge.

For now, we stick with this hypergraph picture. In other words, we identify a source with the set of measurements that it connects to. For the following, we fix a hypergraph $G=(V,E)$, satisfying~\ref{hgp1},~\ref{hgp2}, together with some $d\in\N$ for the number of possible outocmes. We take this data to represent any correlation scenario. We write $V=\{v_1,\ldots,v_n\}$ and associate to each vertex $v_i$ a random variable, representing the measurement outcome distribution, which we also denote by $v_i$. The following definition generalizes the Definitions~\ref{P4corrdef},~\ref{P5corrdef},~\ref{C3corrdef},~\ref{C4corrdef} and~\ref{Skcorrdef}.

\begin{defn}
\label{gencorr}
A \emph{correlation} $p$ in $G$ is a probability distribution $p(v_1,\ldots v_n)$ such that for every pair of subsets $U,W\subseteq V$ which are not connected in $G$ (i.e.~$\not\exists e\in E$ with $U\cap e\neq\emptyset \land W\cap e\neq\emptyset$),
$$
p(u_1,\ldots u_{|U|}, w_1,\ldots,w_{|W|}) = p(u_1,\ldots,u_{|U|}) p(w_1,\ldots,w_{|W|}).
$$
where $U=\{u_1,\ldots,u_{|U|}\}$ and $W=\{w_1,\ldots,w_{|W|}\}$.
\end{defn}

It follows immediately that the same property not only holds for a pair of subsets of $V$, but for any number of pairwise not connected subsets.

\begin{prob}
For every standard Bell scenario, there is a general probabilistic theory~\cite{Barrett} which reproduces all no-signaling correlations in that scenario\footnote{For example, take the corresponding no-signaling polytope as the state space of the total system.}. Is this also true that for every correlation scenario? If not, are there other frameworks beyond general probabilistic theories in which this would be the case? Or would that mean that our Definition~\ref{gencorr} is too lax?
\end{prob}

In a hidden variable model, each source $e\in E$ is described by a hidden variable $\lambda_e$ with some distribution $p(\lambda_e)$. The locality assumption~\ref{L} then allows an outcome $v_i$ to depend on all the sources connected to $v_i$; we write $\Lambda_i=\{\lambda_e; v_i\in e\}$ for the set of hidden variables associated to all those sources.

\begin{defn}
\label{genclassdef}
A correlation $p$ in $G$ is classical if there are distributions $p(\lambda_e)$ and conditional distributions $p(v_i|\Lambda_i)$ such that
\beq
\label{classcor}
p(v_1,\ldots,v_n) = \int_{\{\lambda_e,\:e\in E\}} \prod_{v_i\in V} p(v_i|\Lambda_i) \prod_{e\in E} p(\lambda_e)
\eeq
\end{defn}

See~\ref{genclassdef} for the precise measure-theoretical definition. It is a simple exercise to check that every classical correlation is indeed a correlation as in Definition~\ref{gencorr}.

\begin{prob}
\label{allclass}
Under which conditions on $G$ are all correlations classical?
\end{prob}

A class of scenarios in which all correlations are trivially classical is this:

\begin{prop}
If there is a source in $G$ connecting to all vertices, i.e.~if $E=\{V\}$, then every distribution $p(v_1,\ldots,v_n)$ is a classical correlation in $G$.
\end{prop}

\begin{proof}
The hidden variable carried by the common source can be taken to be $\lambda=(v_1,\ldots,v_n)$ itself: in each run of the experiment, it selects a joint outcome $(v_1,\ldots,v_n)$ according to the desired distribution, sends this joint outcome as a hidden variable $\lambda$ to all measurements. The outcome $v_i$ is then defined to be the $i$th component of $\lambda$. 
\end{proof}

Using our previous analysis of example scenarios together with a bit of graph theory, we can answer Problem~\ref{allclass} at least in the case of bipartite sources, i.e.~when the hypergraph $G=(V,E)$ is a (undirected, simple) graph. The relevant class of correlation scenarions turns out to be the class of \emph{star scenarios} $S_k$ indexed by the number $k\in\N$. The star graph $S_k$ is defined to have vertices $V=\{a,b_1,\ldots,b_k\}$ and one edge between $a$ and every $b_i$, i.e.
$$
E = \{\{a,b_1\},\ldots,\{a,b_k\}\}.
$$
See Figure~\ref{S5fig}.

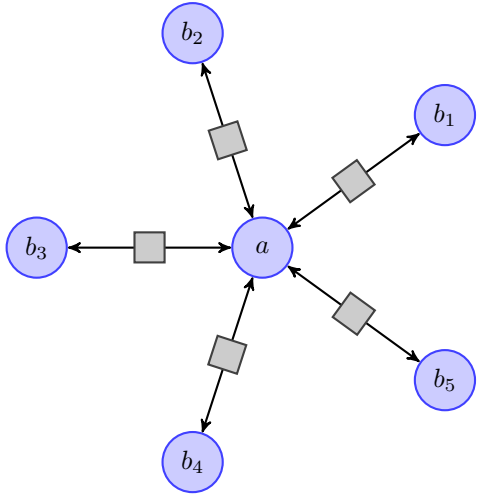
\begin{figure}
\begin{tikzpicture}[node distance=2.1cm,>=stealth',thick]
\tikzstyle{place}=[circle,thick,draw=blue!75,fill=blue!20,minimum size=8mm]
\tikzstyle{transition}=[rectangle,thick,draw=black!75,fill=black!20,minimum size=4.0mm]
\node[place] (a) at (0,0) {$a$} ;
\node[place] (b1) at (36:3) {$b_1$} ;
\node[place] (b2) at (108:3) {$b_2$} ;
\node[place] (b3) at (180:3) {$b_3$} ;
\node[place] (b4) at (252:3) {$b_4$} ;
\node[place] (b5) at (324:3) {$b_5$} ;
\draw[<->] (a) -- (b1) node [midway,sloped,above=-6pt,transition] {} ;
\draw[<->] (a) -- (b2) node [midway,sloped,above=-6pt,transition] {} ;
\draw[<->] (a) -- (b3) node [midway,sloped,above=-6pt,transition] {} ;
\draw[<->] (a) -- (b4) node [midway,sloped,above=-6pt,transition] {} ;
\draw[<->] (a) -- (b5) node [midway,sloped,above=-6pt,transition] {} ;
\end{tikzpicture}
\caption{The correlation scenario $S_5$.}
\label{S5fig}
\end{figure}

\begin{thm}
\label{classgraph}
If the hypergraph $G$ is a graph, then all correlations in $G$ are classical if and only if $G$ is a star graph or a disjoint union of star graphs.
\end{thm}

We begin the proof with a lemma.

\begin{lem}
Let $G=(V,E)$ be a connected simple graph. If $G$ is not a star graph, then $G$ has some induced subgraph which is a $C_3$, $C_4$ or $P_4$.
\end{lem}

\begin{proof}
We use induction on $n=|V|$. For $n\leq 3$, the statement is clear, since the only connected graphs at most three vertices are $C_3$ and the star graphs $P_1=S_0$, $P_2=S_1$ and $P_3=S_2$. For $n\geq 4$, we start with $G$ and assume that $G$ does not contain any induced $C_3$, $C_4$ or $P_4$. We now select any induced subgraph on $n-1$ vertices. By the induction assumption, this subgraph is a star graph with some central vertex $a\in V$ and leaves $b_1,\ldots,b_{n-2}\in V$. For the induction step, we ask: how can the additional vertex $c\in V$ be connected to $a,b_1,\ldots,b_{n-2}$? An edge from $c$ to $a$ together with one from $c$ to some $b_i$ would give rise to an induced subgraph of type $C_3$; no edge to $a$ but an edge to some $b_i$ would give rise to an induced subgraph of type $C_4$ or $P_4$. Therefore, $c$ cannot share an edge with any $b_i$. Then due to connectedness, it needs to share an edge with $a$, which turns it into another leaf of the star.
\end{proof}

\begin{proof}[Proof of Theorem~\ref{classgraph}]
If $G$ is not a star graph or a disjoint union of star graphs, then the lemma guarantees that $G$ contains an induced $C_3$, $C_4$ or $P_4$. Any correlation on such an induced subgraph can be extended to a correlation on $G$ by taking the measurements associated to the additional vertices to have a deterministic outcome. Any hidden variable model of this extension can be restricted to a hidden variable model of the original correlation on the subgraph; in other words, if the original correlation is non-classical, then so is the extension. The existence of non-classical correlations on $G$ now follows from Theorems~\ref{P4thm},~\ref{C3thm} and~\ref{C4thm}.

We now consider the case that $G=(V,E)$ is a star graph. This means that $V=\{a,b_1,\ldots,b_n\}$, where $a$ is the central vertex sharing an edge with each $b_i$, and there are no other edges. It follows from Definition~\ref{gencorr} that a correlation on $G$ is a distribution $p(a,b_1,\ldots,b_n)$ satisfying
$$
p(a,b_1,\ldots,b_n) = p(a|b_1,\ldots,b_n) \prod_{i=1}^n p(b_i) .
$$
Defining hidden variables as $\lambda_{\T{A}\T{B}_i} = b_i$ shows that $p$ is indeed classical.
\end{proof}

Since $C_3$ is the only hypergraph on $3$ vertices for which not all measurements share a common source and which is not a star graph, we obtain as a direct consequence:

\begin{cor}
\label{C3simplest}
$C_3$ is the smallest scenario in which non-classical correlations exist.
\end{cor}

If Problem~\ref{C4prob}\ref{C4qstnq} has a positive answer, then ``non-classical'' can also be replaced in Theorem~\ref{classgraph} and Corollary~\ref{C3simplest} by ``non-classical quantum''.

For Bell scenarios, it is an open problem whether all quantum correlations in a fixed Bell scenario can be achieved quantum-mechanically in terms of quantum states on a Hilbert space of fixed dimension. Numerical evidence suggests that this is not the case in general~\cite{PV}. Due to Theorem~\ref{P4thm}, this question as well as the numerical evidence automatically transfer to the $P_4$ scenario. The analogous question for the classical case is: how many values for the hidden variable(s) are required in order to simulate all classical correlations? In a Bell scenario, this is easily seen to be a finite number since the set of classical correlations is a convex polytope with the deterministic correlations as extremal points, so that Carath{\'e}odory's Theorem gives an explicit bound on the number of hidden variable values needed. However, in our more general formalism, the answer to the same question is not at all clear.

\begin{prob}
\label{hvfin}
Are there correlation scenarios in which no finite number of values for the hidden variables is enough for obtaining all classical correlations with a given number of outcomes?
\end{prob}

Due to Theorem~\ref{P4thm}, we know that a finite number is sufficient in the case of $P_4$. The natural next step will be to consider this problem for $C_3$, where it already seems difficult. 

\begin{prob}
\label{finpi}
Can the set of classical correlations be described by a finite number of polynomial inequalities?
\end{prob}

This is in fact related to Problem~\ref{hvfin}:

\begin{prop}
Let $G$ be a correlation scenario with a fixed number of outcomes for each measurement. If a finite number of hidden variable value suffices in $G$ to obtain all classical correlations, then the set of classical correlations in $G$ can be described in terms of a finite number of polynomial inequalities.
\end{prop}

\begin{proof}
If $k\in\N$ hidden variable values are enough to simulate all classical correlations, then a distribution over these values is specified by $k-1$ real numbers satisfying $k$ linear inequalities. Similarly, a conditional distribution $p(v_i|\Lambda_i)$ is specified by a certain finite number of real variables satisfying certain linear inequalities. The question of whether a given correlation is classical is then equivalent to asking whether these real variables can be chosen in such a way that they satisfy these linear inequalities and reproduce the given $p$ via\eq{classcor}. In other words, it boils down to deciding whether a given system of polynomial inequalities, containing the $p(v_1,\ldots,v_n)$ as parameters, has a solution over $\R$.

Thanks to Tarski's real quantifier elimination~\cite{Tarski}, this system of polynomial inequalities is solvable if and only if $p(v_1,\ldots,v_n)$ itself satisfies certain polynomial inequalities which can in principle be computed explicitly.
\end{proof}

Besides the trivial case of star graph scenarios, the only cases for which we know a positive answer to Problem~\ref{hvfin}, and therefore Problem~\ref{finpi}, are the $P_4$ scenario (Theorem~\ref{P4thm} and~\cite{Fine}) and the $P_5$ scenario (Theorem~\ref{P5thm} and~\cite{BGP}).

We have already noted in Remark~\ref{notconvex} that the set of classical correlations is not convex in general. So one may wonder:

\begin{prob}
What is the shape of the set of classical correlations? Can it have a non-trivial topology, or is it always homeomorphic to a ball of the appropriate dimension? If yes, what is this dimension? If no, is the set nevertheless contractible, or can it have ``holes''? Is it always simply connected? What about the analogous questions for the set of quantum correlations?
\end{prob}

At the moment, we can only offer a very simple observation concerning these topological questions:

\begin{prop}
Let $G=(V,E)$ be any correlation scenario with the number of outcomes of each measurement fixed to some $d\in\N$. Then the set of classical correlations is path-connected.
\end{prop}

\begin{proof}
Given classical correlations $p_0(v_1,\ldots,v_n)$ and $p_1(v_1,\ldots,v_n)$ on $G$, we describe how to construct an explicit $1$-parameter family of correlations continuously interpolating between these two. The assumption of classicality means that there are hidden variable distributions $p(\lambda_1^0),\:\ldots,\: p(\lambda_m^0)$ and $p(\lambda^1_1),\:\ldots,\: p(\lambda^1_m)$ together with the appropriate conditional distributions $p(v_i|\Lambda^0_i)$ and $p(v_i|\Lambda_i^1)$ such that
$$
p_0(v_1,\ldots,v_n) = \int_{\{\lambda_e;\: e\in E\}} \prod_{v_i\in V} p(v_i|\Lambda_i^0) \prod_{e\in E} p(\lambda_e^0) ,
$$
$$
p_1(v_1,\ldots,v_n) = \int_{\{\ell_e;\: e\in E\}} \prod_{v_i\in V} p(v_i|\Lambda_i^1) \prod_{e\in E} p(\lambda_e^1) ,
$$
We now define a continuous family of classical correlations indexed by a parameter $t\in[0,1]$. These us as hidden variables the pairs $\lambda_e = (\lambda_e^0,\lambda_e^1)$ with distribution $p(\lambda_e) = p(\lambda_e^0,\lambda_e^1) \eqdef p(\lambda_e^0) p(\lambda_e^1)$.

For every $t\in[0,1]$, we define a new conditional distribution for each random variable $v_i$,
\beq
\label{conddeform}
p_t(v_i|\Lambda_i) = (1-t)\cdot p(v_i | \Lambda_i^0 ) + t\cdot p(v_i | \Lambda_i^1 ) ,
\eeq
and consider the resulting joint distribution
\beq
\label{corrdeform}
p_t(v_1,\ldots,v_n) = \int_{\{\lambda_e;\: e\in E\}} \prod_{i\in V} p_t(v_i|\Lambda_i) \prod_{e\in E} p(\lambda_e) .
\eeq
By construction, this is a family of classical correlations depending continuously on $t$. For $t=0$, the conditional distributions\eq{conddeform} do not depend on the $\lambda_e^1$ component of $\lambda_e=(\lambda_e^0,\lambda_e^1)$, so that the integration over $\lambda_e^1$ in\eq{corrdeform} is trivial and the original $p_0(v_1,\ldots,v_n)$ is reproduced. Similarly for $t=1$. Then by continuity in $t$, the family $p_t$ defines a continuous path of classical correlations between the two given classical correlations $p_0$ and $p_1$.
\end{proof}

For a similar proof idea, see~\cite{BRGP}*{App.~A.1}.

We now return to the original picture of Definition~\ref{2layerdef} and consider some generalities on quantum correlations, starting with the rigorous definition. For a Hilbert space $\H$, we write $S(\H)$ for the set of states on $\H$, i.e.~positive trace-class operators of unit trace norm.

\begin{defn}
\label{gendefq}
Let $G=(S,M,C,d)$ be a correlation scenario. A correlation $p(v_1,\ldots,v_n)$ in $G$ is \emph{quantum} if the following data exist:
\begin{enumerate}
\item for every connection $(s,m)\in C$, a Hilbert space $\H_{(s,m)}$;
\item for every source $s\in S$, a quantum state $\rho_s \in\S\left(\bigotimes_{\{m\::\:sCm\}}\H_{(s,m)}\right)$;
\item for every measurement $m\in M$, a POVM $\{\mathcal{F}_m\}$ with elements $\mathcal{F}_m\in \B\left(\bigotimes_{\{s\::\:sCm\}}\H_{(s,m)}\right)$;
\end{enumerate}
such that
\beq
\label{rigdefeq}
p(m_1,\ldots,m_n) = \mathrm{tr}\left[ \left( \bigotimes_{s\in S} \rho_s \right) \left( \bigotimes_{v_i\in V} \mathcal{F}_{v_i} \right) \right]
\eeq
\end{defn}

In this equation, both the left as well as the right tensor product evaluate to operators on $\bigotimes_{(s,m)\in C} \H_{(s,m)}$, but the two tensor products are taken with respect to different orders on $C$. We take it is as understood that these tensor products are taken to be reordered in such a way that the corresponding factors match.

We leave it to the reader to show that every classical correlation is also quantum.

\begin{prop}
\label{sepclass}
If all sources in a correlation scenario emit separable quantum states, then the resulting correlation is classical.
\end{prop}

\begin{proof}
Here, we assume all Hilbert spaces to be finite-dimensional; see~\ref{gensepproof} for the general case.

Carath{\'e}odory's Theorem guarantees the existence of some number $k\in\N$ such that every $\rho_s$ can be decomposed as
\beq
\label{sepstates}
\rho_s = \sum_{j=1}^k \mu_{s,j} \bigotimes_{\{m\::\:sCm\}} \rho_{(s,m,j)}
\eeq
for certain coefficients $\mu_{s,j}\geq 0$ with $\sum_{j=1}^k \mu_{s,j}=1$ and certain states $\rho_{(s,m,j)}\in S\left(\H_{(s,m)}\right)$. For each source $s$, we define its hidden variable $\lambda_s$ to take values $j_s\in\{1,\ldots,k\}$ with distribution $p(\lambda_s=j_s) = \mu_{s,j}$ and
\beq
\label{scm}
p(m | \lambda_s = j_s \textrm{ for all } s \textrm{ with } sCm ) \equiv \mathrm{tr}\left[ \bigotimes_{s\::\: sCm} \rho_{(s,m,j_s)} \mathcal{F}_{m} \right] .
\eeq
This reproduces the correlation\eq{rigdefeq} for the states\eq{sepstates}. Instead of verifying this formally, we would like to mention its interpretation as a concrete physical protocol. According to the decomposition\eq{sepstates}, each source $s$ can produce its state $\rho_s$ by randomly generating $\lambda_s$, distribution according to the weights $\mu_{s,j}$, and preparing and sending the corresponding state $\rho_{(s,m,j)}$ to each party $m$ for which $sCm$. In order to turn this into a completely classical protocol, we may shift the preparation of the states $\rho_{(s,m,j)}$ from the sources to the parties: if each party $m$ knows the values of the hidden variables $\lambda_s$ for all $s$ with $sCm$, then this party itself can prepare the required states $\rho_{(s,m,j)}$ locally and measure them. In this way, only classical information $\lambda_s$ has to be sent from the sources to the parties, and the parties' preparation and measurement can be considered as a single classical measurement on the $\lambda_s$'s given by the conditional probabilities\eq{scm}.
\end{proof}

\begin{prob}
Does every entangled quantum state display non-classical quantum correlations? I.e.~can one obtain non-classical quantum correlations by choosing an appropriate correlation scenario and putting one copy of the state in each source? Does it help if each source also emits classical shared randomness in addition to the entangled state?
\end{prob}

\appendix

\section{Measure-theoretical technicalities and other nuisances}
\label{app}

In the main text, we have assumed all our hidden variables to be discrete for the sake of readability. We drop this assumption here and consider the most general case: hidden variables can be arbitrary probability spaces. The following subsection are all referenced from the main text, so this appendix should be referred to only as needed.

\subsection{What is a hidden variable?} 
\label{wihv}

The literature knows examples of discrete hidden variables and continuous hidden variables. In standard Bell scenarios, Carath{\'e}odory's Theorem guarantees that considering discrete hidden variables is enough; unfortunately, we do not know whether this also holds for our case (Problem~\ref{hvfin}). Therefore, we should allow hidden variables which are as general as possible and require a definition which not only comprises discrete and continuous hidden variables, but also allos intermediate possibilities and even hidden variable with \emph{more than} continuously many values.

Since the only successful general theory of (classical) randomness is the one based on the Kolmogorov axioms for probability measures and probability spaces, this is what seems to us to be the only reasonable general definition of hidden variable:

\begin{udefn}
A \emph{hidden variable} is a probability space $(\Omega,\mathcal{E},P)$.
\end{udefn}

We think of the actual value of the hidden variable to be a ranodm element $\lambda\in\Omega$ with distribution $P$. This is the most general kind of classical hidden variable we can imagine. It comprises both discrete and continuous variables as special cases as well as everything else, for example hidden variables with so many values that $\Omega$ has cardinality greater than the continuum.

\subsection{Distributions conditional on hidden variables}
\label{appcond}

Definitions~\ref{P4classdef},~\ref{P5classdef},~\ref{C3classdef},~\ref{C4classdef},~\ref{genclassdef} talk about outcome distributions conditional on one or several hidden variables. What does a conditional distribution, like $p(a|\lambda)$, mean when $\lambda$ is not discrete?

There are several equivalent ways to answer this question. We have chosen the following one which is convenient in that it is partly formulated in terms familiar from quantum theory.

\begin{udefn}
Let $L^\infty(\Omega,\mathcal{E},P)$ be the von Neumann algebra associated to $(\Omega,\mathcal{E},P)$. A \emph{distribution of $a$ conditional on $\lambda\in\Omega$} is an assignment of some positive operator $\mathcal{O}_a^*=\mathcal{O}_a\in L^\infty(\Omega,\mathcal{E},P)$, $\mathcal{O}_a\geq 0$, to every $a$ such that $\sum_{a} \mathcal{O}_a = \mathbbm{1}$. 
\end{udefn}

The attentive reader will have noticed that this is nothing but a POVM in $L^\infty(\Omega,\mathcal{E},P)$ indexed by $a$. Roughly speaking, each $\mathcal{O}_a$ is a real-valued function on $\Omega$ whose values $\mathcal{O}_a(\lambda)$ represent the conditional probabilities $p(a|\lambda)$. For finite $\Omega$ with $\mathcal{E}=2^\Omega$ and $P({\lambda})>0$ for every $\lambda\in\Omega$, this intuition is exact; in general though, it has to be kept in mind that $\mathcal{O}_a$ is not a single function, but rather a whole equivalence class of functions, such that expressions like
$$
p(a) = \int_{\lambda} \mathcal{O}_a(\lambda) P(\lambda)
$$
are well-defined in the sense that the value of the integral is independent of the choice of representative.


In general, a measurement $a$ will depend on several hidden variables given by probability spaces $(\Omega_1,\mathcal{E}_1,P_1),\:\ldots,\:(\Omega_n,\mathcal{E}_n,P_n)$. In this case, $\mathcal{O}$ should be a POVM in the von Neumann algebra of the product probability space $\left(\prod_i\Omega_i,\prod_i\mathcal{E}_i,\prod_i P_i\right)$.

We now state Definition~\ref{genclassdef} again in the present language.

\begin{udefn}
\label{rigdef}
Let $G=(V,E)$ be a correlation scenario. A correlation $p(v_1,\ldots,v_n)$ in $G$ is \emph{classical} if the following data exist:
\begin{enumerate}
\item for every $e\in E$, a hidden variable $\lambda_e$ given in terms of a probability space $(\Omega_e,\mathcal{E}_e,P_e)$;
\item Conditional probabilities $\mathcal{O}_a\in L^\infty(\Omega_{\Lambda_i},\mathcal{E}_{\Lambda_i},P_{\Lambda_i})$ where $\Lambda_i=\{\lambda_e; v_i\in e\}$ is the collection of hidden variables associated to all the sources connected to $v_i$, and $(\Omega_{\Lambda_i},\mathcal{E}_{\Lambda_i},P_{\Lambda_i})$ is the corresponding product probability space;
\end{enumerate}
such that
\beq
\label{rigdefeq}
p(v_1,\ldots,v_n) = \int_{\{\lambda_e;\: e\in E\}} \prod_{v_i\in V} \mathcal{O}_{v_i}(\Lambda_i) \prod_{e\in E} dP(\lambda_e)
\eeq
\end{udefn}

In particular, this clarifies also the definitions of classical correlation in our example scenarios, Definitions~\ref{P4classdef}~\ref{P5classdef},~\ref{C3classdef},~\ref{C4classdef}

\subsection{Hidden variables can be assumed deterministic}
\label{appwlogdet}

We have outlined in Remark~\ref{wlogdet} why the conditional distributions $\mathcal{O}_a$ as used above can in fact taken to be deterministic. In our present picture, determinism means $\mathcal{O}_a^2=\mathcal{O}_2$, i.e.~that $\mathcal{O}_a$ is a projection. This is equivalent to $\mathcal{O}_a(\lambda)\in\{0,1\}$ for almost all $\lambda\in\Omega$ which corresponds to determinism in the form $p(a|\lambda)\in\{0,1\}$.

We now turn the intuitive argument of Remark~\ref{wlogdet} into a rigorous proof sketch.

\begin{uprop}
Let $G=(V,E)$ be a correlation scenario. If $p$ is classical, then there exists a classical model for $p$ in which all $\mathcal{O}_{v_i}$ are projections.
\end{uprop}

\begin{proof}
We show how to replace the $\mathcal{O}_{w}$'s by a projection for some fixed $w\in V$; the claim then follows from applying this procedure to every vertex $w\in V$. We start by choosing a source $e\in E$ which connects to $v_1$ and replace the given probability space $(\Omega_e,\mathcal{E}_e,P_e)$ by $\Omega'_e\equiv \Omega_e\times[0,1]$, which we take to be equipped with the product $\sigma$-algebra $\mathcal{E}'_e$ and the product measure $P'_e$, where $[0,1]$ carries the Lebesgue $\sigma$-algebra and measure; the second factor in this product represents the additional random number mentioned in Remark~\ref{wlogdet}. We enumerate the possible outcomes as $w\in\{1,\ldots,d\}$ for some $d\in\N$, and define
$$
\mathcal{O}'_{w} \;:\; \Omega_e\times [0,1] \to \{0,1\},\quad (\lambda,x) \mapsto \left\{ \begin{array}{cl} 1 & \textrm{if } \sum_{w'=1}^{w-1} \mathcal{O}_{w'}(\lambda) \leq x < \sum_{w'=1}^{w} \mathcal{O}_{w'}(\lambda) \\ 0 & \textrm{otherwise} \end{array}\right.
$$
which is easily seen to represent a projection in $L^\infty(\Omega'_e,\mathcal{E}'_e,P'_e)$. The requirement $\sum_{w=1}^d \mathcal{O}'_w = \mathbbm{1}$ holds by construction in $L^\infty(\Omega'_e,\mathcal{E}'_e,P'_e)$, i.e.~up to a set of measure zero.

All $\mathcal{O}_{v_i}$ with $v_i\neq w$ connecting to $e$ we take to operate as before in the sense that we replace them by $\mathcal{O}'_{v_i}(\lambda,x) = \mathcal{O}_{v_i}(\lambda)$; all other sources $\neq e$ and all measurements not connected to $e$ remain completely unchanged. 

We leave it to the reader to verify that these replacements preserve the correlation. 
\end{proof}

\subsection{General proof of Lemma~\ref{C3ei}}
\label{eigenproof}

We follow essentially the same lines as in the discrete-variable proof of the main text. Since we do not know of a formulation of the data processing inequality for (relative) Shannon entropy on arbitrary probability spaces, and similarly for submodularity of entropy, we make our own definitions and derive our inequalities in analogy with the discrete case. We start with the first argument involving the data processing inequality. In order to obtain finite quantitites, we need to work with conditional entropies, in which the hidden variables appear only as conditioning variables. For the sake of illustration, we start with the discrete-variable case, in which
$$
H(a|\lambda_{\T{A}\T{B}}) = \sum_{a,\lambda_{\T{A}\T{B}}} f\left(p(a|\lambda_{\T{A}\T{B}})\right) p(\lambda_{\T{A}\T{B}}),
$$
where we abbreviated $f(x)=-x\cdot \log x$, with $f(0)\equiv 0$ as usual. Thanks to the condtional independence $p(a|\lambda_{\T{A}\T{B}})=p(a|\lambda_{\T{A}\T{B}},b)$ and concavity of $f$,
\begin{align*}
H(a|\lambda_{\T{A}\T{B}}) &= \sum_{a,b,\lambda_{\T{A}\T{B}}} f\left(p(a|\lambda_{\T{A}\T{B}})\right) p(\lambda_{\T{A}\T{B}})p(b|\lambda_{\T{A}\T{B}}) = \sum_{a,b} \sum_{\lambda_{\T{A}\T{B}}} f\left(p(a|\lambda_{\T{A}\T{B}},b)\right) p(\lambda_{\T{A}\T{B}}|b) p(b) \\
&\leq \sum_{a,b} f\left(\sum_{\lambda_{\T{A}\T{B}}} p(a|\lambda_{\T{A}\T{B}},b) p(\lambda_{\T{A}\T{B}}|b)\right) p(b) = \sum_{a,b} f\left( p(a|b)\right) p(b) = H(a|b).
\end{align*}
We now emulate this estimate in the general case by defining
$$
H(a|\lambda_{\T{A}\T{B}}) \equiv \sum_a \int_{\lambda_{\T{A}\T{B}}} f\left( \mathcal{O}_a(\lambda_{\T{A}\T{B}}) \right) dP_{\T{A}\T{B}}(\lambda_{\T{A}\T{B}})
$$
and noting that this is well-defined, thanks to $\mathcal{O}_a(\lambda_{\T{A}\T{B}})\in [0,1]$ a.s., and coincides with the standard definition in the discrete case. We rewrite this as
$$
H(a|\lambda_{\T{A}\T{B}}) = \sum_{a,b} \int_{\lambda_{\T{A}\T{B}}} f\left( \mathcal{O}_a(\lambda_{\T{A}\T{B}}) \right) \frac{\mathcal{O}_b(\lambda_{\T{A}\T{B}}) dP_{\T{A}\T{B}}(\lambda_{\T{A}\T{B}})}{p(b)} p(b) .
$$
Now for $p(b)>0$, the fraction in the integrand is again a measure on $(\Omega_{\T{A}\T{B}},\mathcal{E}_{\T{A}\T{B}})$, and Jensen's inequality gives
$$
H(a|\lambda_{\T{A}\T{B}}) \leq \sum_{a,b} f\left( \int_{\lambda_{\T{A}\T{B}}} \mathcal{O}_a(\lambda_{\T{A}\T{B}}) \cdot \frac{\mathcal{O}_b(\lambda_{\T{A}\T{B}}) dP_{\T{A}\T{B}}(\lambda_{\T{A}\T{B}})}{p(b)} \right) p(b)
$$
Since $p(a,b) = \int_{\lambda_{\T{A}\T{B}}} \mathcal{O}_a(\lambda_{\T{A}\T{B}})\mathcal{O}_b(\lambda_{\T{A}\T{B}}) dP_{\T{A}\T{B}}$, the integral inside $f$ evaluates to $p(a|b)$, so that
\beq
\label{ing1}
H(a|\lambda_{\T{A}\T{B}}) \leq \sum_{a,b} f\left( p(a|b) \right) p(b) = H(a|b),
\eeq
which is the data processing inequality we wanted to prove.

We now make the usual estimates known from proofs of nonnegativity of conditional mutual information or nonnegativity of Kullback-Leibler divergence~\cite{CT}*{Thm.~8.6.1},
\begin{align*}
& H(a|\lambda_{\T{A}\T{B}}) + H(a|\lambda_{\T{A}\T{C}}) - H(a) - H(a|\lambda_{\T{A}\T{B}}\lambda_{\T{A}\T{C}})  \\
& = \sum_a \left[ \int_{\lambda_{\T{A}\T{B}},\lambda_{\T{A}\T{C}}} \mathcal{O}_a(\lambda_{\T{A}\T{B}},\lambda_{\T{A}\T{C}}) \bigg( - \log(\mathcal{O}_a(\lambda_{\T{A}\T{B}})) - \log(\mathcal{O}_a(\lambda_{\T{A}\T{C}})) + \log(p(a)) + \log(\mathcal{O}_a(\lambda_{\T{A}\T{B}},\lambda_{\T{A}\T{C}})) \bigg)\, dP_{\T{A}\T{B}}\, dP_{\T{A}\T{C}} \right] \\
& = - \sum_a \left[ \int_{\lambda_{\T{A}\T{B}},\lambda_{\T{A}\T{C}}} \mathcal{O}_a(\lambda_{\T{A}\T{B}},\lambda_{\T{A}\T{C}}) \log \left( \frac{\mathcal{O}_a(\lambda_{\T{A}\T{B}})\mathcal{O}_a(\lambda_{\T{A}\T{C}})}{p(a)\mathcal{O}_a(\lambda_{\T{A}\T{B}},\lambda_{\T{A}\T{C}})} \right) \, dP_{\T{A}\T{B}}\, dP_{\T{A}\T{C}} \right] \\
& \geq - \log \left[ \sum_a \int_{\lambda_{\T{A}\T{B}},\lambda_{\T{A}\T{C}}} \mathcal{O}_a(\lambda_{\T{A}\T{B}},\lambda_{\T{A}\T{C}}) \cdot \frac{\mathcal{O}_a(\lambda_{\T{A}\T{B}})\mathcal{O}_a(\lambda_{\T{A}\T{C}})}{p(a)\mathcal{O}_a(\lambda_{\T{A}\T{B}},\lambda_{\T{A}\T{C}})} dP_{\T{A}\T{B}}\, dP_{\T{A}\T{C}} \right] = - \log\left[ \sum_a \frac{p(a)p(a)}{p(a)} \right] = 0.
\end{align*}
Since $H(a|\lambda_{\T{A}\T{B}}\lambda_{\T{A}\T{C}})$ is defined as the integral of an a.s.~nonnegative function, it is itself nonnegative, and therefore
\beq
\label{ing2}
H(a|\lambda_{\T{A}\T{B}}) + H(a|\lambda_{\T{A}\T{C}}) \geq H(a) .
\eeq
Piecing finally the two ingredients\eq{ing1} and\eq{ing2} together, we find
$$
I(a:b) + I(a:c) = 2H(a) - H(a|b) - H(a|c) \leq 2H(a) - H(a|\lambda_{\T{A}\T{B}}) - H(a|\lambda_{\T{A}\T{C}}) \leq H(a),
$$
as was to be shown.

\subsection{General proof of Theorem~\ref{C4thm}}
\label{C4proof}

In the discrete-variable case, we started with the assumption that the measurements were deterministic and noticed that if a certain combination of outcomes has positive probability, then there has to be a combination of hidden variable values, each occurring with positive probability, which produces that outcome combination.

This reasoning needs to be modified in order to apply in the general case; when dealing with non-atomic probability spaces, no single hidden variable combination has positive probability. It is therefore necessary to consider combinations of \emph{sets} of hidden variable values, which is unfortunately somewhat technical.

\begin{ulem}
Let $(\Omega_1,\mathcal{E}_1,P_1),\:\ldots,\:(\Omega_n,\mathcal{E}_n,P_n)$ be probability spaces and let $\Omega=\prod_{i=1}^n\Omega_i$ be equipped with the product $\sigma$-algebra $\mathcal{E}=\sigma\left(\prod_{i=1}^n \mathcal{E}_i\right)$ and the product measure $P=\prod_{i=1}^n P_i$, so that $(\Omega,\mathcal{E},P)$ is a probability space.

Then, for a measurable function $f:\Omega\to\{0,1\}$ with $P(f=1)>0$ and any $\eps>0$, there exist measurable subsets $\Xi_i\subseteq\Omega_i$, with $P_i(\Xi_i)>0$, such that\\
$$
P\left(\,f=1\,|\,\Xi_1\times\ldots\times\Xi_n\,\right) > 1 - \eps.
$$
\end{ulem}

\begin{proof}
This lemma can be reformulated as saying that if $\Theta\subseteq\Omega$ has positive measure, then there exist $\Xi_i\subseteq\Theta$ of positive measure such that $P(\Theta|\prod_{i=1}^n \Xi_i) > 1-\eps$.

We start to prove this reformulation by noting that the collection of sets which are finite disjoint unions of product sets is an algebra of sets~\cite{Halmos}*{33.E}. It then follows from the approximation lemma of measure theory~\cite{Halmos}*{13.D} that $\Theta$, a set of positive measure, can be $\delta$-approximated by a set $S(\delta)$ which is a finite union of product sets, i.e.~for every $\delta>0$ we can find such $S(\delta)$ with $P(\Theta\setminus S(\delta))<\delta$ and $P(S(\delta)\setminus\Theta)<\delta$. We assume $\delta<P(\Theta)$, so that $P(S(\delta))>0$ is guaranteed.

Decomposing this $S(\delta)$ into a finite union of disjoint product sets gives
$$
S(\delta) = \bigcup_{j=1}^{k(\delta)} \Xi^j(\delta)
$$
for product sets $\Xi^j(\delta)=\Xi_1^j(\delta)\times\ldots\times\Xi_n^j(\delta)$, which we may assume to be of positive measure (if some $\Xi^j(\delta)$ has zero measure, then it may as well be omitted). By construction, we know
$$
\sum_{j=1}^{k(\delta)} P(\Xi^j(\delta)\cap \Theta) > P(\Theta) - \delta,\qquad \sum_{j=1}^{k(\delta)} P(\Xi^j(\delta)\setminus\Theta) < \delta.
$$
Since the second inequality states that
$$
\sum_{j=1}^k \frac{P(\Xi^j(\delta)\cap\Theta)}{P(S(\delta))} \cdot \frac{P(\Xi^j(\delta)\setminus\Theta)}{P(\Xi^j(\delta)\cap\Theta)} < \frac{\delta}{P(S(\delta))} ,
$$
and this sum is a convex combination, we conclude that there is at least one index $j$ for which
$$
\frac{P(\Xi^j(\delta)\setminus\Theta)}{P(\Xi^j(\delta)\cap\Theta)} < \frac{\delta}{P(S(\delta))} < \frac{\delta}{P(\Theta)-\delta} .
$$
We define $\Xi = \prod_{i=1}^n \Xi_i$ to be equal to this $\Xi^j(\delta)$. Then
$$
P(\Theta|\Xi) = \frac{P(\Theta\cap\Xi)}{P(\Xi\setminus\Theta)+P(\Xi\cap\Theta)} = \left(1 + \frac{P(\Xi\setminus\Theta)}{P(\Xi\cap\Theta)}\right)^{-1} > \left(1 + \frac{\delta}{P(\Theta)-\delta}\right)^{-1}.
$$
For $\delta$ sufficiently small, this is $>1-\eps$, as has been claimed.
\end{proof}

We return to the main line of the proof of Theorem~\ref{C4thm} and fix $\eps>0$. In a hidden variable combination like $(\ell_{\T{A}\T{B}},\ell_{\T{B}\T{Y}},\ell_{\T{Y}\T{X}},\ell_{\T{X}\T{A}})$, each component now becomes a \emph{set} of hidden variable values having positive probability. By the lemma, we can choose these sets in such a way that this when such a combination of hidden variables occurs, then the joint outcome is $(0,0,0,0)$ with probability $>1-\eps$. In particular, when a hidden variable combination in $(\ell_{\T{A}\T{B}},\ell_{\T{B}\T{Y}},\ast,\ast)$ occurs, where the last two components are unspecified, then $b=0$ with probability $>1-\eps$. Similarly, we find a combination of sets $\kappa_{\T{A}\T{B}},\kappa_{\T{B}\T{Y}},\kappa_{\T{Y}\T{X}},\kappa_{\T{X}\T{A}})$ producing $(1,0,1,1)$ with probability $>1-\eps$. Therefore, the combination $(\kappa_{\T{A}\T{B}},\kappa_{\T{B}\T{Y}},\ell_{\T{Y}\T{X}},\kappa_{\T{X}\T{A}})$ yields $(1,0,1,1)$ with probability $>1-2\eps$; it should now be clear how to complete the proof, following the steps of the discrete-variable case and bounding the probabilities in each step. Choosing $\eps$ small enough then shows that the probability to get the outcome $(0,0,1,1)$ is strictly positive in contradiction with\eq{PRbox}.

\subsection{Separable states give rise to classical correlations}
\label{gensepproof}

Here, we lift the restriction of finite-dimensionality from the proof of Proposition~\ref{sepclass}. First of all, what does separability even mean in the infinite-dimensional case? In the following, we work with arbitrary Hilbert spaces $\H$ which are not necessarily separable, and put the usual trace-norm topology on $S(\H)$; upon interpreting a quantum state on $\H$ as a normal positive linear functional on $\B(\H)$, this is the weak $*$-topology. Moreover, $S(\H)$ carries the Borel $\sigma$-algebra induced from this (metrizable) topology.

\begin{udefn}[cf.~\cite{HSW}]
Let $\H_1,\ldots,\H_k$ be Hilbert spaces. A state $\rho\in S(\H_1\otimes\ldots\otimes\H_k)$ is \emph{separable} if it lies in the closed convex hull of the set of product states.
\end{udefn}

In general, one cannot expect a separable state to have a decomposition into a finite or infinite convex combination of product states; rather, integrals are needed~\cite{HSW}.

\begin{ulem}
Let $\rho\in S(\H_1\otimes\ldots\otimes\H_k)$ be separable. Then there exists a probability measure $P$ on the set of product states such that
$$
\rho = \int_{S(\H_1)\otimes\ldots\otimes S(H_k)} \left(\rho_1\otimes\ldots \otimes\rho_k \right) dP(\rho_1\otimes\ldots\otimes\rho_k)
$$
\end{ulem}

In the finite-dimensional case, one can take the measure $P$ to have finite support, so that the integral becomes a finite convex combination.

\begin{proof}
Since the set of product states is compact, Milman's converse to the Krein-Milman Theorem guarantess that every extreme point of the set of separable states is a product state. Then the assertion follows from Choquet's Theorem~\cite{Phelps}.
\end{proof}

This should make it clear how to prove Proposition~\ref{sepclass} in the general case: to a source $s$ sending out a separable state
$$
\rho_s = \int_{\prod_{\{m\::\: sCm\}} S(\H_{(s,m)})} \left( \bigotimes_{\{m\::\: sCm\}} \rho_{(s,m)} \right) dP_s\left(\bigotimes_{\{m\::\: sCm\}} \rho_{(s,m)}\right) ,
$$
we associate the hidden variable probability space $\Omega_s=\prod_{\{m\::\: sCm\}} S(\H_{(s,m)})$ equipped with its Borel $\sigma$-algebra and its probability measure $P_s$, so that the hidden variable $\lambda_s$ ranges over all product states $\lambda_s=\bigotimes_{\{m\::\: sCm\}} \rho_{(s,m)}$. Concerning the conditional probabilities,\eq{scm} now reads
$$
\mathcal{O}_{m} \left(\left\{\lambda_s = \bigotimes_{\{m'\::\: sCm'\}} \rho_{(s,m')} : sCm\right\}\right) \equiv \mathrm{tr}\left[ \left(\bigotimes_{s\::\: sCm} \rho_{(s,m)}\right) \mathcal{F}_{m} \right] .
$$
This is a continuous, and therefore measurable, function on $\Omega_s=\prod_{\{m\::\: sCm\}} S(\H_{(s,m)})$.

The intuition about how this classical model works is similar to the finite-dimensional case. One may think of the hidden variable $\lambda_s$ as an abstract classical description of the product state sent out by the source; each party $m$ then receives all descriptions from all the sources it connects to, for each of these products states retains only the information concerning the system to him while throwing away the rest, and uses that information to calculate his required outcome distribution which can then get sampled in order to obtain his outcome. By construction, this produces the desired joint distribution of outcomes.

\end{document}